\documentclass{article}
\usepackage[height=9in,width=6.2in]{geometry}
\usepackage[utf8]{inputenc}
\usepackage{natbib}
\usepackage{graphicx}
\usepackage{amsmath}
\usepackage{amssymb,amsthm}
\usepackage{bm}
\usepackage{paralist}
\usepackage[lined,boxed]{myalg2e}
\usepackage{slashbox}

\newcommand{\Sp}{\mathbb{S}}
\newcommand{\Spd}{\Sp^{p-1}}
\newcommand{\reals}{\mathbb{R}}

\newcommand{\x}{\bm{x}}
\newcommand{\m}{\bm{\mu}}

\newcommand{\kummer}{M}
\newcommand{\hyperpq}[2]{\thinspace{}_{{#1}}F_{{#2}}}

\newcommand{\risingF}[2]{{{#1}}^{\overline{{#2}}}}
\newcommand{\half}{\tfrac{1}{2}}
\newcommand{\khat}{\hat{\kappa}}

\newcommand{\set}[1]{\left\{{#1}\right\}}

\newcommand{\X}{{\cal X}}

\newcommand{\ratio}{\frac{M'(a, c; \kappa)}{M(a, c; \kappa)}}
\newcommand{\mh}{\hat{\m}}
\newcommand{\proj}{\mathbb{P}^{p-1}}
\newcommand{\fmu}{f_\mu}
\newcommand{\nlsum}{\sum\nolimits}
\numberwithin{equation}{section}

\DeclareMathOperator*{\argmax}{argmax}

\newtheorem{theorem}{Theorem}[section]
\newtheorem{lemma}[theorem]{Lemma}

\newcommand{\email}[1]{\texttt{#1}}

\begin{document}
\title{The Multivariate Watson Distribution: Maximum-Likelihood Estimation and other Aspects}
\author{\begin{tabular}{ccc}
    Suvrit Sra\thanks{This work was largely done when the first author was affiliated with the Max Planck Institute for Biological Cybernetics} &\hskip 12pt &Dmitrii Karp\\
    Max Planck Institute for Intelligent Systems & & Institute of Applied Mathematics\\
    T\"ubingen, Germany & & Vladivostok, Russian Federation\\
    \email{suvrit@tuebingen.mpg.de} & & \email{dimkrp@gmail.com}
  \end{tabular}}
\date{}
\maketitle

\begin{abstract}
  This paper studies fundamental aspects of modelling data using
  multivariate Watson distributions. Although these distributions
  are natural for modelling axially symmetric data (i.e., unit vectors
  where $\pm \x$ are equivalent),  for high-dimensions using them can
  be difficult---largely because for  Watson distributions even basic
  tasks such as maximum-likelihood are  numerically challenging. To tackle
  the numerical difficulties some approximations have been derived. But
  these are either grossly inaccurate in high-dimensions
  (\emph{Directional Statistics}, Mardia \& Jupp. 2000)
  or  when reasonably accurate (\emph{J. Machine Learning Research,
  W.\& C.P., v2},  Bijral \emph{et al.}, 2007, pp. 35--42), they lack
  theoretical justification.  We derive asymptotically precise two-sided
  bounds for the maximum-likelihood estimates which lead to new approximations.
  Our approximations are theoretically  well-defined, numerically accurate,
  and easy to compute. We build  on our parameter estimation and discuss
  mixture-modelling with Watson distributions; here we uncover a hitherto  unknown connection
  to the ``diametrical clustering'' algorithm of Dhillon \emph{et al.}
  (\emph{Bioinformatics}, 19(13), 2003, pp. 1612--1619).
\end{abstract}
\begin{flushleft}
\end{flushleft}

\section{Introduction}
Life on the surface of the unit hypersphere is more twisted than
you might imagine: designing elegant probabilistic models is easy
but using them is often not. This difficulty usually stems from
the complicated normalising constants associated with directional
distributions. Nevertheless, owing to their powerful modelling
capabilities, distributions on hyperspheres continue finding
numerous applications---see e.g., the excellent book
\emph{Directional Statistics}~\citep{maju00}.

A fundamental directional distribution is the von Mises-Fisher
(vMF) distribution, which models data concentrated around a
mean-direction. But for data that have additional structure, vMF
can be inappropriate: in particular, for axially symmetric data it
is more natural to prefer the (Dimroth-Scheidegger)-Watson
distribution~\citep{maju00,watson65}. And this distribution is the
focus of our paper.

Three main reasons motivate our study of the multivariate Watson
(mW) distribution, namely:
\begin{inparaenum}[(i)]
\item it is fundamental to directional statistics;
\item it has not received much attention in modern data-analysis setups involving high-dimensional data; and
\item it provides a theoretical basis to ``diametrical clustering'', a
  procedure developed for gene-expression analysis~\citep{diametric}.
\end{inparaenum}

Somewhat surprisingly, for high-dimensional settings, the mW
distribution seems to be fairly under-studied. One reason might be
that the traditional domains of directional statistics are
low-dimensional, e.g., circles or spheres. Moreover, in
low-dimensions numerical difficulties that are rife in
high-dimensions are not so pronounced. This paper contributes
theoretically and numerically to the study of the mW distribution.
We hope that these contributions and the connections we make to
established applications help promote wider use of the mW
distribution.

\subsection{Related Work}
Beyond their use in typical applications of directional
statistics~\citep{maju00}, directional distributions gained
renewed attention in data-mining, where the vMF distribution was
first used by~\citep{bdgs03:kdd,suv.vmf}, who also derived some
\emph{ad-hoc} parameter estimates; Non \emph{ad-hoc} parameter
estimates for the vMF case were obtained by~\citet{tanabe}.

More recently, the Watson distribution was considered
in~\citep{avleen} and also in~\citep{suv.phd}. \citet{avleen} used
an approach similar to that of~\citep{bdgs03:kdd} to obtain an
\emph{ad-hoc} approximation to the maximum-likelihood estimates.
We eliminate the \emph{ad-hoc} approach and formally derive tight,
two-sided bounds which lead to parameter approximations that are
accurate and efficiently computed.

Our derivations are based on carefully exploiting properties
(several \emph{new} ones are derived in this paper) of the
confluent hypergeometric function, which arises as a part of the
normalisation constant. Consequently, a large body of classical
work on special functions is related to our paper. But to avoid
detracting from the main message limitations, we relegate highly
technical details to the appendix.

Another line of related work is based on mixture-modelling with
directional distributions, especially for high-dimensional
datasets. In~\citep{suv.vmf}, mixture-modelling using the
Expectation Maximisation (EM) algorithm for mixtures of vMFs was
related to cosine-similarity based K-means clustering.
Specifically, \citet{suv.vmf} showed how the cosine based K-means
algorithm may be viewed as a limiting case of the EM algorithm for
mixtures of vMFs. Similarly, we investigate mixture-modelling
using Watson distributions, and connect a limiting case of the
corresponding EM procedure to a clustering algorithm called
``diametrical clustering''~\citep{diametric}. Our viewpoint
provides a new interpretation of the (discriminative) diametrical
clustering algorithm and also lends generative semantics to it.
Consequently, using a mixture of Watson distributions we also
obtain a clustering procedure that can provide better clustering
results than plain diametrical clustering alone.

\vspace*{-12pt}
\section{Background}
Let $\Spd = \{\bm{x} \mid \bm{x} \in \reals^p,\ \|\bm{x}\|_2 =
1\}$ be the $(p{-}1)$-dimensional unit hypersphere centred at the
origin.  We focus on axially symmetric vectors, i.e., $\pm\x \in
\Spd$ are equivalent; this is also denoted by $\x \in \proj$,
where $\proj$ is the projective hyperplane of dimension $p{-}1$. A
natural choice for modelling such data is the multivariate Watson
distribution~\citep{maju00}. This distribution is parametrised by
a \emph{mean-direction} $\bm\mu \in \proj$, and a
\emph{concentration} parameter $\kappa \in \reals$; its
probability density function is
\begin{equation}
  \label{eq:60}
  W_p(\x; \m, \kappa) = c_p(\kappa) e^{\kappa (\bm{\mu}^T\bm{x})^2},\qquad \x \in \proj.
\end{equation}
The normalisation constant $c_p(\kappa)$ in~\eqref{eq:60} is given
by
\begin{equation}
  \label{eq:one}
  c_p(\kappa) = \frac{\Gamma(p/2)}{2\pi^{p/2}\kummer(\half, \tfrac{p}{2}, \kappa)},
\end{equation}
where $\kummer$ is Kummer's confluent hypergeometric function
defined as (\citep[formula~6.1(1)]{htf} or
\cite[formula~(2.1.2)]{AAR}
\begin{equation}
  \label{eq:two}
  \kummer(a, c, \kappa) = \sum_{j \geq 0} \frac{\risingF{a}{j}}{\risingF{c}{j}} \frac{\kappa^j}{j!},\qquad a, c, \kappa \in \reals,
\end{equation}
and $\risingF{a}{0}=1$, $\risingF{a}{j}=a(a+1)\cdots(a+j-1)$,
$j\geq{1}$, denotes the \emph{rising-factorial}.

Observe that for $\kappa > 0$, the density concentrates around
$\bm{\mu}$ as $\kappa$ increases, whereas for $\kappa < 0$, it
concentrates around the great circle orthogonal to $\bm{\mu}$.
Observe that $(\bm{Q\mu})^T\bm{Qx} = \bm{\mu}^T\bm{x}$ for any
orthogonal matrix $\bm{Q}$. In particular for $\bm{Q\mu}=\m$,
$\m^T(\bm{Qx})=\m^T\bm{x}$; thus, the Watson density is
rotationally symmetric about $\bm{\mu}$.

\subsection{Maximum Likelihood Estimation}
\label{sec:mle} We now consider the basic and apparently simple
task of maximum-likelihood parameter estimation for mW
distributions: this task turns out to be surprisingly difficult.

Let $\x_1,\ldots,\x_n \in \proj$ be i.i.d.\ points drawn from
$W_p(\x;\m,\kappa)$, the Watson density with mean $\m$ and
concentration $\kappa$. The corresponding log-likelihood is
\begin{equation}
  \label{eq:3}
  \ell(\m,\kappa; \x_1,\ldots,\x_n) = n\bigl(\kappa\m^T\bm{S}\m -
  \log\kummer(1/2, p/2, \kappa) + \gamma\bigr),
\end{equation}
where $\bm{S} = n^{-1}\sum_{i=1}^n \x_i\x_i^T$ is the sample
\emph{scatter
  matrix}, and $\gamma$ is a constant term that we can ignore. Maximising~\eqref{eq:3} leads to the following parameter estimates~\citep[Sec.~10.3.2]{maju00} for the mean vector
\begin{equation}
  \label{four}
  \hat{\m} = \bm{s}_1\quad\text{if}\quad \khat > 0,\qquad \mh = \bm{s}_p\quad\text{if}\quad \khat < 0,
\end{equation}
where $\bm{s}_1,\ldots,\bm{s}_p$ are normalised eigenvectors ($\in
\proj$) of the scatter matrix $\bm S$ corresponding to the
eigenvalues $\lambda_1 \ge \lambda_2 \ge \cdots \ge \lambda_p$.
The concentration estimate $\khat$ is obtained by
solving\footnote{More precisely, we need $\lambda_1 > \lambda_2$
to ensure a unique m.l.e.\ for positive $\kappa$, and
$\lambda_{p-1} > \lambda_p$, for negative $\kappa$.}
\begin{equation}
  \label{eq:five}
  g(\half,\tfrac{p}{2}; \khat) := \frac{\kummer'(\half, \tfrac{p}{2}, \khat)}
  {\kummer(\half, \tfrac{p}{2}, \khat)}\ \ =\ \ \mh^T\bm{S}\mh\ :=\
  r\qquad(0\le r\le 1),
\end{equation}
where $M'$ denotes the derivative with respect to $\khat$. Notice
that~\eqref{four} and~\eqref{eq:five} are coupled---so we need
some way to decide whether to solve $g(1/2,p/2;\khat) = \lambda_1$
or to solve $g(1/2, p/2;\khat) =\lambda_p$ instead. An easy choice
is to solve both equations, and select the solution that yields a
higher log-likelihood. Solving these equations is much harder.

One could solve~\eqref{eq:five} using a root-finding method
(e.g.~Newton-Raphson). But, the situation is not that simple. For
reasons that will soon become clear, an out-of-the-box
root-finding approach can be unduly slow or even fraught with
numerical peril, effects that become more pronounced with
increasing data dimensionality. Let us, therefore, consider a
slightly more general equation (we also drop the accent on
$\kappa$):
\begin{center}
  \fbox{
    \begin{minipage}[h]{0.8\linewidth}
      \begin{center}
        \textbf{Solve for $\kappa$}
      \end{center}
      \begin{equation}
        \label{eq:61}
        \begin{split}
        \quad g(a, c; \kappa) := &\ratio = r\\
        c > a > 0,&\quad 0 \leq r \leq 1.
        \end{split}
      \end{equation}
    \end{minipage}
  }
\end{center}

\section{Solving for $\kappa$}
\label{sec:analysis} In this section we present two different
solutions to~\eqref{eq:61}. The first is the ``obvious'' method
based on a Newton-Raphson root-finder. The second method is the
key numerical contribution of this paper: a method that computes a
closed-form approximate solution to~\eqref{eq:61}, thereby
requiring merely a few floating-point operations!

\subsection{Newton-Raphson}
Although we establish this fact not until
Section~\ref{sec:direct}, suppose for the moment
that~\eqref{eq:61} \emph{does} have a solution.  Further, assume
that by bisection or otherwise, we have bracketed the root
$\kappa$ to be within an interval and are thus ready to invoke the
Newton-Raphson method.

Starting at $\kappa_0$, Newton-Raphson solves the equation $g(a,c;
\kappa)-r=0$ by iterating
\begin{equation}
  \label{eq:33}
  \kappa_{n+1} = \kappa_n - \frac{g(a, c; \kappa_n)-r}{g'(a, c; \kappa_n)},\qquad n = 0,1,\ldots.
\end{equation}
This iteration may be simplified by rewriting $g'(a, c; \kappa)$.
First note that
\begin{equation}
  \label{eq:62}
  g'(a, c; \kappa) = \frac{M''(a, c; \kappa)}{M(a, b; \kappa)} - \biggl(\ratio\biggr)^2,
\end{equation}
then, recall the following two identities
\begin{align}
  \label{eq:40}
  M''(a, c; \kappa)\quad&=\quad\frac{a(a+1)}{c(c+1)}M(a+2,c+2; \kappa);\\
  \label{eq:63}
  M(a+2,c+2;\kappa)\quad&=\quad\frac{(c+1)(-c+\kappa)}{(a+1)\kappa}M(a+1,c+1;\kappa)
  + \frac{(c+1)c}{(a+1)\kappa}M(a,c; \kappa).
\end{align}
Now, use both~\eqref{eq:40} and~\eqref{eq:63} to rewrite the
derivative~\eqref{eq:62} as
\begin{equation}
 \label{eq:41}
   g'(a, c; \kappa) = (1-c/\kappa)g(a, c; \kappa) + (a/\kappa) - (g(a, c;\kappa))^2.
\end{equation}
The main consequence of these simplifications is that
iteration~(\ref{eq:33}) can be implemented with only \emph{one}
evaluation of the ratio $g(a, c; \kappa_n) =
M'(a,c;\kappa_n)/M(a,c;\kappa_n)$. Efficiently computing this
ratio is a \emph{non-trivial} task in itself; an insight into this
difficulty is offered by observations in~\citep{gautschi77,gil}.
In the worst case, one may have to compute the numerator and
denominator separately (using multi-precision floating point
arithmetic), and then divide. Doing so can require several million
extended precision floating point operations, which is very
undesirable.

\subsection{Closed-form Approximation for~\eqref{eq:61}}
\label{sec:direct} We now derive two-sided bounds which will lead
to a closed-form approximation to the solution of~(\ref{eq:61}).
This approximation, while marginally less accurate than the one
via Newton-Raphson, should suffice for most uses. Moreover, it is
incomparably faster to compute as it is in closed-form.

Before proceeding to the details, let us look at a little history.
For 2--3 dimensional data, or under very restrictive assumptions
on $\kappa$ or $r$, some approximations had been previously
obtained~\citep{maju00}. Due to their restrictive assumptions,
these approximations have limited applicability, especially for
high-dimensional data, where these assumptions are often
violated~\citep{suv.vmf}.  Recently~\citet{avleen} followed the
technique of~\cite{suv.vmf} to essentially obtain the
\emph{ad-hoc} approximation (actually particularly for the case
$a=1/2$)
\begin{equation}
  \label{eq:bijral}
  BBG(r) := \frac{cr-a}{r(1-r)} + \frac{r}{2c(1-r)},
\end{equation}
which they observed to be quite accurate. However,
\eqref{eq:bijral} lacks theoretical justification; other
approximations were presented in~\citep{suv.phd}, though again
only \emph{ad-hoc}.

Below we present new approximations for $\kappa$ that are
theoretically well-motivated and also numerically more accurate.
Key to obtaining these approximations are a set of bounds
localizing $\kappa$, and we present these in a series of theorems
below. The proofs are given in the appendix.

\subsubsection{Existence and Uniqueness}
The following theorem shows that the function $g(a, c; \kappa)$ is
strictly increasing.
\begin{theorem}\label{cor:exist}
Let $c>a>0$, and $\kappa \in \reals$. Then the function $\kappa\to
g(a, c; \kappa)$ is monotone
  increasing from $g(a,c;-\infty)=0$ to $g(a,c;\infty)= 1$.
\end{theorem}
\begin{proof}
 Since $g(a,c;\kappa)=(a/c)f_1(\kappa)$, where $f_{\mu}$ is defined
 in (\ref{eq.1}), this theorem is a direct consequence of
 Theorem~\ref{thm:monotone}.\qedhere
\end{proof}
Hence the equation $g(a, c; \kappa)=r$ has a unique solution for
each $0 < r < 1$. This solution is negative if $0 < r < a/c$ and
positive if $a/c < r < 1$. Let us now localize this solution to a
narrow interval by deriving tight bounds on it.

\subsubsection{Bounds on the solution $\kappa$}
Deriving tight bounds for $\kappa$ is key to obtaining our new
theoretically well-defined numerical approximations; moreover,
these approximations are easy to compute because the bounds are
given in closed form.

\begin{theorem}
  \label{thm:bounds}
  Let the solution to $g(a, c; \kappa)=r$ be denoted by $\kappa(r)$. Consider the following three bounds:
  \begin{align}
    \label{eq:lower}
    &(\text{lower bound})\qquad &L(r)&=\frac{rc-a}{r(1-r)}\left(1+\frac{1-r}{c-a}\right),\\
    \label{eq:mid}
    &(\text{bound})\qquad &B(r)&=\frac{rc-a}{2r(1-r)}\left(1+\sqrt{1+\frac{4(c+1)r(1-r)}{a(c-a)}}\right),\\
    \label{eq:upper}
    &(\text{upper bound})\qquad &U(r)&=\frac{rc-a}{r(1-r)}\left(1+\frac{r}{a}\right).
  \end{align}
  Let $c > a > 0$, and $\kappa(r)$ be the solution~\eqref{eq:61}. Then, we have
  \begin{enumerate}
  \item for $a/c<r<1$,
    \begin{equation}\label{eq:newpos}
      L(r) < \kappa(r) < B(r) < U(r),
    \end{equation}
  \item for $0<r<a/c$,
    \begin{equation}\label{eq:newneg}
      L(r) < B(r) < \kappa(r) < U(r).
    \end{equation}
  \item and if $r=a/c$, then $\kappa(r) = L(a/c)=B(a/c)=U(a/c)=0$.
  \end{enumerate}
  All three bounds ($L$, $B$, and $U$) are also
  asymptotically precise at $r=0$ and $r=1$.
\end{theorem}
\begin{proof}
  The proofs of parts 1 and~2 are given in Theorems~\ref{thm.poskappa}
  and \ref{thm.negkappa} (see Appendix), respectively. Part~3 is trivial.
  It is easy to see that $\lim_{r \to{0,1}}U(r)/L(r)=1$, so
  from inequalities~\eqref{eq:newpos} and~\eqref{eq:newneg}, it follows
  that
  \begin{equation*}
    \lim\limits_{r\to{0,1}}\frac{L(r)}{\kappa(r)} = \lim\limits_{r\to{0,1}}\frac{B(r)}{\kappa(r)} = \lim\limits_{r\to{0,1}}\frac{U(r)}{\kappa(r)}=1.\qedhere
  \end{equation*}
\end{proof}
More precise asymptotic characterizations of the approximations
$L$, $B$ and $U$ are given in section~\ref{sec.asymp}.

\subsubsection{BBG approximation}
Our bounds above also provide some insight into the previous
heuristically motivated $\kappa$-approximation $BBG(r)$
of~\citet{avleen} given by~(\ref{eq:bijral}). Specifically, we
check whether $BBG(r)$ satisfies the lower and upper bounds from
Theorem~\ref{thm:bounds}.

To see when $BBG(r)$ violates the lower bound, solve $L(r)>BBG(r)$
for $r$ to obtain
$$
\frac{2c^2+a-\sqrt{(2c^2-a)(2c^2-a-8ac)}}{2(2c^2-a+c)}<r<\frac{2c^2+a+\sqrt{(2c^2-a)(2c^2-a-8ac)}}{2(2c^2-a+c)}.
$$
For the Watson case $a=1/2$; this means that $BBG(r)$ violates the
lower bound and \emph{underestimates} the solution for
$r\in(0.11,0.81)$ if $c=5$; for $r\in(0.0528,0.904)$ if $c=10$;
for $r\in(0.00503,0.99)$ if $c=100$; for $r\in(0.00050025,0.999)$
if $c=1000$.  This fact is also reflected in Figure~\ref{fig.one}.

To see when $BBG(r)$ violates the upper bound, solve $BBG(r)>U(r)$
for $r$ to obtain
$$
r<\frac{2ac}{2c^2-a}.
$$
For the Watson case $a=1/2$; this means that $BBG(r)$ violates the
upper bound and \emph{overestimates} the solution for
$r\in(0,0.1)$ if $c=5$; for $r\in(0,0.05)$ if $c=10$; for
$r\in(0,0.005)$ if $c=100$; for $r\in(0,0.0005)$ if $c=1000$.

What do these violations imply? They show that a combination of
$L(r)$ and $U(r)$ is guaranteed to give a better approximation
than $BBG(r)$ for nearly all $r\in(0,1)$ except for a very small
neighbourhood of the point where $BBG(r)$ intersects $\kappa(r)$.

\subsubsection{Asymptotic precision of the approximations}
\label{sec.asymp} Let us now look more precisely at how the
various approximations behave at limiting values of $r$. There are
three points where we can compute asymptotics: $r=0$, $r=a/c$, and
$r=1$. First, we assess how $\kappa(r)$ itself behaves.
\begin{theorem}
  Let $c > a > 0$, $r \in (0,1)$; let $\kappa(r)$ be the solution to $g(a, c; \kappa)=r$. Then,
  \begin{align}
    \label{eq:kappa0}
    \kappa(r) &= -\frac{a}{r}+(c-a-1)+\frac{(c-a-1)(1+a)}{a}r+O(r^2),\quad r\to{0},\\
    \label{eq:kappaac}
    \kappa(r) &= \Bigl(r-\frac{a}{c}\Bigr) \left  \{ \frac{c^2(1+c)}{a(c-a)} + \frac{c^3(1+c)^2(2a-c)}{a^2(c-a)^2(c+2)} \Bigl(r-\frac{a}{c}\Bigr) + O\Bigl(\Bigl(r-\frac{a}{c}\Bigr)^2\Bigr) \right\},\quad r\to{\frac{a}{c}}\\
\label{eq:kappa1}
    \kappa(r) &= \frac{c-a}{1-r} + 1-a + \frac{(a-1)(a-c-1)}{c-a}(1-r) + O((1-r)^2),\quad r\to{1}.
  \end{align}
\end{theorem}
This theorem is given in the appendix with detailed proof as
Theorem~\ref{th:xasymp}.

We can compute asymptotic expansions for the various
approximations by standard Laurent expansion. For $L(r)$ we
obtain:
$$
L(r)=-\frac{a(c-a+1)}{(c-a)r}+(c/(c-a)+c-a)+O(r),\quad r\to{0},
$$
$$
L(r)=\frac{c^2(c+1)}{a(c-a)}(r-a/c)+\frac{c^3(ac-(c+1)(c-a))}{a^2(c-a)^2}(r-a/c)^2+O((r-a/c)^3),\quad
r\to{a/c},
$$
$$
L(r)=\frac{c-a}{1-r}+(1-a)+\frac{a(a-c-1)}{c-a}(1-r)+O((1-r)^2),\quad
r\to{1}.
$$
For $U(r)$ we get:
$$
U(r)=-\frac{a}{r}+(c-a-1)+\frac{(c-a)(a+1)}{a}r+O(r^2),\quad
r\to{0},
$$
$$
U(r)=\frac{c^2(c+1)}{a(c-a)}(r-a/c)+\frac{c^3(2ac+a-c^2)}{a^2(c-a)^2}(r-a/c)^2+O((r-a/c)^3),\quad
r\to{a/c},
$$
$$
U(r)=\frac{((c/a)-a+c-1)}{1-r} - ((c/a)+a)+O(1-r),\quad r\to{1}.
$$
For $B(r)$ the expansions are:
$$
B(r)=-\frac{a}{r} + \frac{a^2-2ac+c^2-c-1}{c-a}+O(r),\quad
r\to{0},
$$
$$
B(r)=\frac{c^2(1+c)}{a(c-a)}(r-a/c)+\frac{c^3(1+c)^2(2a-c)}{a^2(c-a)^2(c+2)}(r-a/c)^2+O((r-a/c)^3),\quad
r\to{a/c},
$$
$$
B(r)=\frac{c-a}{1-r} + \frac{c+1-a^2}{a}+O(1-r),\quad r\to{1}.
$$
Finally, for the approximation (\ref{eq:bijral}) we have the
following expansions:
$$
BBG(r)=-\frac{a}{r}+(c-a)+O(r),\quad r\to{0},
$$
$$
BBG(r)=\frac{a}{2c(c-a)}+O(r-a/c),\quad r\to{a/c},
$$
$$
BBG(r)=\frac{2ac-2c^2-1}{2c(1-r)} - \frac{2ac+1}{2c}+O(1-r),\quad
r\to{1}.
$$
We summarize the results in Table~\ref{tab.summ} below.

Table~\ref{tab.summ} uses the following terminology: (i) we call
an approximation $f(r)$ to be \emph{incorrect} around $r=\alpha$,
if $f(r)/\kappa(r) \to {0,\infty}$ as $r\to \alpha$; (ii) we say
$f(r)$ is \emph{correct of order~1} around $r=\alpha$, if
$f(r)/\kappa(r) \to{C}$ such that $C \neq 0, \infty$ as
$r\to\alpha$; (iii) we say $f(r)$ is \emph{correct of order~2}
around $r=\alpha$ if
$f(r)/\kappa(r)=1+O(r-\alpha)$ as $r\to\alpha$; and (iv) $f(r)$ is \emph{correct of order~3} around $r=\alpha$ if $f(r)/\kappa(r)=1+O((r-\alpha)^2)$ as $r\to\alpha$. 

No matter how we count the total ''order of correctness'' it is
clear from Table~\ref{tab.summ} that our approximations are
superior to that
of~\citep{avleen}.%

The table shows that actually $L(r)$ and $U(r)$ can be viewed as
three-point $[2/2]$ Pad\'{e} approximations to $\kappa(r)$ at
$r=0$ and $r=a/c$ and $r=1$ with different orders at different
points, while $B(r)$ is a special non-rational three point
approximation with even higher total order of contact.

Moreover, since we not only give the order of correctness but also
prove the inequalities, we always know exactly which approximation
underestimates $\kappa(r)$ and which overestimates $\kappa(r)$.
Such information might be important to some applications. The
approximation of~\citep{avleen} is clearly less precise and does
not satisfy such inequalities. Also, note that all the above facts
are equally true in the Watson case $a=1/2$.

\begin{table}[ht]\small
  \centering
  \renewcommand{\arraystretch}{2}
  \begin{tabular}{|c|c|c|c|c|}
    \hline
    \backslashbox[10pt]{Point}{Approx.}  &  $L(r)$ & $B(r)$ & $U(r)$ & $BBG(r)$\\
    \hline
    $r=0$     & \parbox{0.6in}{Correct of order 1} & \parbox{0.6in}{Correct of order 2}& \parbox{0.6in}{Correct of  order 3} & \parbox{0.6in}{Correct of order 2}\\[5pt]
    \hline
    $r=a/c$   & \parbox{0.6in}{Correct of order 2} & \parbox{0.6in}{Correct of order 3}& \parbox{0.6in}{Correct of  order 2} & \parbox{0.6in}{Incorrect} \\[5pt]
    \hline
    $r=1$     & \parbox{0.6in}{Correct of order 3} & \parbox{0.6in}{Correct of order 2}& \parbox{0.6in}{Correct of  order 1} & \parbox{0.6in}{Correct of order 1} \\[5pt]
    \hline
  \end{tabular}
  \caption{Summary of various approximations}
\label{tab.summ}
\end{table}

\section{Application to Mixture Modelling and Clustering}
\label{sec:mixture} Now that we have shown how to compute
maximum-likelihood parameter estimates, we proceed onto
\emph{mixture-modelling} for mW distributions.

Suppose we observe the set ${\cal X} = \{\bm{x}_1,\dots,\bm{x}_n
\in \proj\}$ of i.i.d.\ samples. We wish to model this set using a
mixture of $K$ mW distributions. Let $W_p(\x | \m_j, \kappa_j)$ be
the density of the $j$-th mixture component, and $\pi_j$ its prior
({\small $1\le j \le K$}) -- then, for observation $\bm{x}_i$ we
have the density
\begin{equation*}
  f(\bm{x}_i | \m_1,\kappa_1,\ldots,\m_K,\kappa_K) =
  \sum\nolimits_{j=1}^K \pi_j W_p(\bm{x}_i | \m_j,\kappa_j).
\end{equation*}
The corresponding log-likelihood for the entire dataset $\cal X$
is given by
\begin{equation}
  \label{eq:17}
  {\cal L}({\cal X}; \m_1,\kappa_1,\ldots,\m_K,\kappa_K) =
  \nlsum_{i=1}^n \log \Bigl(\nlsum_{j=1}^K \pi_j W_p(\bm{x}_i |  \m_j,\kappa_j)\Bigr).
\end{equation}
To maximise the log-likelihood, we follow a standard Expectation
Maximisation (EM) procedure~\citep{dlr77}. To that end, first
bound ${\cal L}$ from below as
\begin{equation}
  \label{eq:19}
  {\cal L}({\cal X}; \m_1,\kappa_1,\ldots,\m_K,\kappa_K) \geq \nlsum_{ij} \beta_{ij} \log \frac{\pi_j W_p(\bm{x}_i |  \m_j, \kappa_j)}{\beta_{ij}},
\end{equation}
where $\beta_{ij}$ is the \emph{posterior} probability (for
$\x_i$, given component $j$), and it is defined by the
\emph{E-Step}:
\begin{equation}
  \label{eq:18}
  \beta_{ij} = \frac{\pi_j W_p(\bm{x}_i | \m_j,\kappa_j)}
  {\sum_l \pi_l W_p(\bm{x}_i | \m_l, \kappa_l)}.
\end{equation}
Maximising the lower-bound~\eqref{eq:19} subject to
$\m_j^T\m_j=1$, yields the \emph{M-Step}:
\begin{align}
  \label{eq:5}
  &\bm{\mu}_j = \bm{s}_1^j\quad\text{if}\quad \kappa_j >
  0,\qquad\bm{\mu}_j = \bm{s}_p^j\quad\text{if}\quad\kappa_j < 0,\\
  \label{eq:6}
  &\kappa_j = g^{-1}(1/2,p/2,r_j),\quad\text{where}\quad
  r_j = \m_j^T\bm{S}^j\m_j,\\
  \nonumber
  &\pi_j = \frac{1}{n}\nlsum_{i} \beta_{ij},
\end{align}
where $\bm{s}_i^j$ denotes the eigenvector corresponding to
eigenvalue $\lambda_i$ (where $\lambda_1 \ge \cdots \ge
\lambda_p$) of the \emph{weighted-scatter matrix:}
\begin{equation*}
  \bm{S}^j = \frac{1}{\nlsum_i \beta_{ij}}\nlsum_i \beta_{ij} \bm{x}_i\bm{x}_i^T.
\end{equation*}
Now we can iterate between~\eqref{eq:18},~\eqref{eq:5}
and~\eqref{eq:6} to obtain an EM algorithm. Pseudo-code for such a
procedure is shown below as Algorithm~\ref{algo:mow}.

\paragraph{Note: Hard Assignments.} We note that as usual, to reduce the computational burden, we can replace can E-step~\eqref{eq:18} by the standard \emph{hard-assignment} heuristic:
\begin{equation}
  \label{eq:31}
  \beta_{ij} =
  \begin{cases}
    1, & \text{if}\ j = \argmax_{j'} \log \pi_{j'} + \log W_p(\x_i| \m_{j'},\kappa_{j'}),\\
    0, & \text{otherwise}.
  \end{cases}
\end{equation}
The corresponding $M$-Step also simplifies considerably. Such
hard-assignments maximize a lower-bound on the incomplete
log-likelihood, and yield \emph{partitional-clustering} algorithms
(in fact, we show experimental results in
Section~\ref{sec:clustering} where we cluster data using a
partitional-clustering algorithm based on this hard-assignment
heuristic).

\begin{algorithm}[t]\small
  \DontPrintSemicolon
  \caption{\small EM Algorithm for mixture of Watson (moW)~\label{algo:mow}}
  \KwIn{$\X = \set{\x_1,\ldots,\x_n  : \text{ where each } \x_i \in \proj}$, $K$: number of components}
  \KwOut{Parameter estimates $\pi_j$, $\m_j$, and $\kappa_j$, for $1 \le j \le K$}
  Initialise $\pi_j, \m_j, \kappa_j$ for $1 \leq j \leq K$\;
  \While{not converged}{
    \emph{\{Perform the E-step of EM\}}\;
    \ForEach{$i$ and $j$}{
      Compute $\beta_{ij}$ using~(\ref{eq:18}) (or via~\eqref{eq:31}
      if using hard-assignments)
    }
    \emph{\{Perform the M-step of EM\}}\;
    \For{$j = 1$ to $K$}{
      $\pi_j \gets \frac{1}{n}\sum_{i=1}^n \beta_{ij}$\;
      Compute $\m_j$ using~(\ref{eq:5})\;
      Compute $\kappa_j$ using~(\ref{eq:6})\;
    }
  }
\end{algorithm}
\begin{algorithm}[!h]\small
  \caption{\small Diametrical Clustering~\label{algo:diametric}}
  \SetAlgoLined
  \DontPrintSemicolon
    \KwIn{${\cal X} = \{\x_1,\ldots,\x_n : \x_i \in \proj\}$, $K$: number of clusters}
    \KwOut{A partition $\{{\cal X}_j : 1 \le j \le K\}$ of ${\cal X}$, and centroids $\m_j$}
    Initialise $\m_j$ for $1 \leq j \leq K$\;
    \While{\emph{not converged}}{
      \emph{E-step:}\;
      Set ${\cal X}_j \gets \emptyset$ for $1 \leq j \leq K$\;
      \For{$i = 1$ to $n$}{
        ${\cal X}_j \gets {\cal X}_j \cup \{\bm{x}_i\}$ where $j = \argmax_{1\le h\le K} (\bm{x}_i^T\m_{h})^2$}
      \emph{M-step:}\;
        \For{$j = 1$ to $K$}{
      $\bm{A}_j = \nlsum_{\x_i \in {\cal X}_j} \bm{x}_i\x_i^T$\;
      $\m_j \gets \bm{A}_j\m_j / \|\bm{A}_j\m_j\|$\;
    }
  }
\end{algorithm}
\begin{figure}
  \centering
\hspace*{-20pt}\includegraphics[scale=0.45]{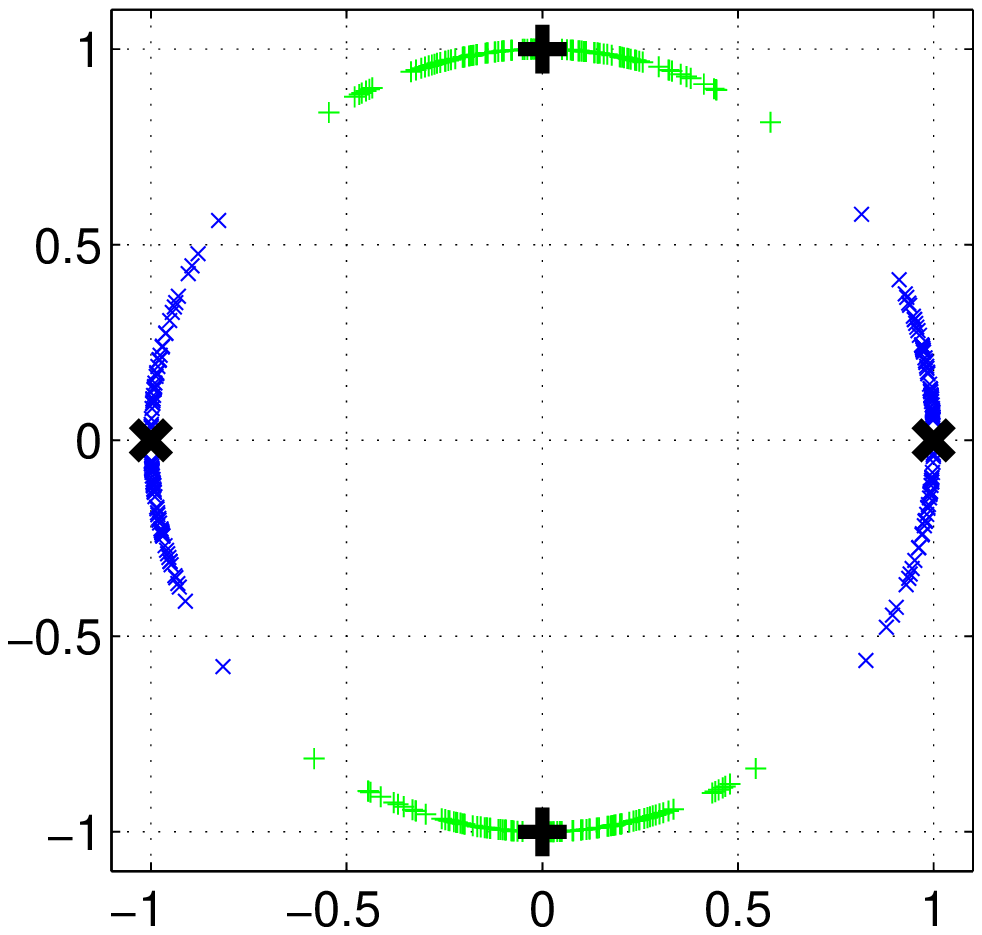}~\includegraphics[scale=0.45]{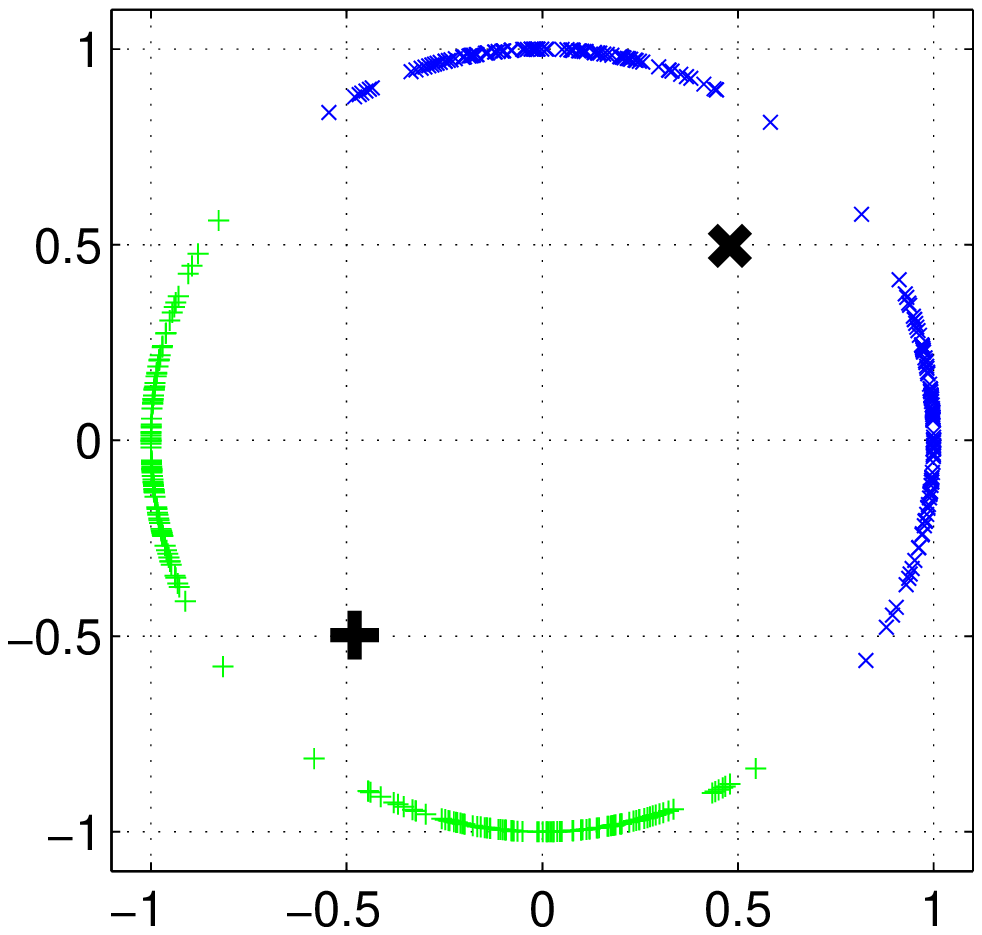}~\includegraphics[scale=0.45]{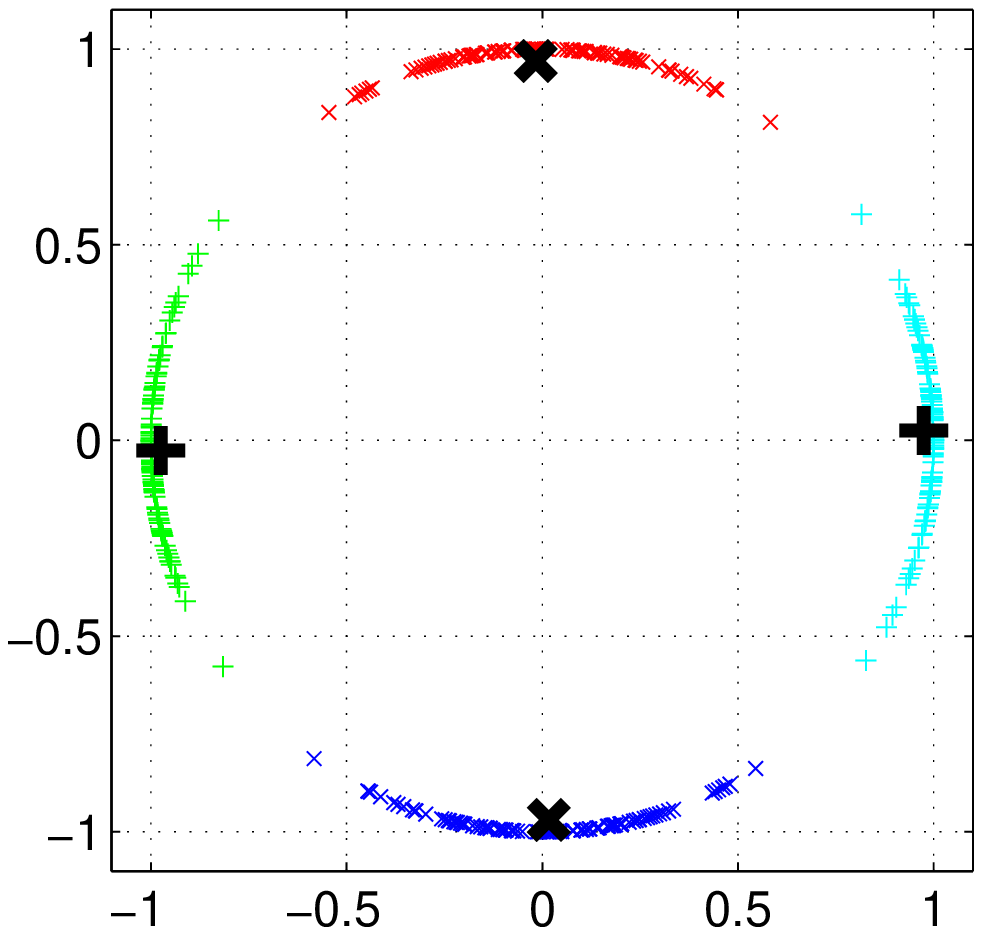}
  \caption{The left panel shows axially symmetric data that has two clusters (centroids are indicated by '+' and  'x'). The middle and right panel shows clustering yielded by (Euclidean) K-means (note that the centroids fail to lie on the circle in this case) with $K=2$ and $K=4$, respectively. Diametrical clustering recovers the true clusters in the left panel.}
  \label{fig:diam}
\end{figure}

\subsection{Diametrical Clustering}
\label{ref:diametric} We now turn to the diametrical clustering
algorithm of~\citet{diametric}, and show that it is merely a
special case of the mixture-model described above. Diametrical
clustering is motivated by the need to group together correlated
and anti-correlated data points (see Figure~\ref{fig:diam} for an
illustration). For data normalised to have unit euclidean norm,
such clustering treats \emph{diametrically} opposite points
equivalently. In other words, $\bm{x}$ lies on the projective
plane. Therefore, a natural question is whether diametrical
clustering is related to Watson distributions, and if so, how?

The answer to this question will become apparent once we recall
the diametrical clustering algorithm (shown as
Algorithm~\ref{algo:diametric}) of~\citep{diametric}.  In
Algorithm~\ref{algo:diametric} we have labelled the ``E-Step'' and
the ``M-Step''. These two steps are simplified instances of the
E-step~(\ref{eq:18}) (alternatively~\ref{eq:31}) and
M-step~(\ref{eq:5}). To see why, consider the
E-step~\eqref{eq:18}. If $\kappa_j \to \infty$, then for each $i$,
the corresponding posterior probabilities~$\beta_{ij} \to
\{0,1\}$; the particular $\beta_{ij}$ that tends to 1 is the one
for which $(\m_j^T\x_i)^2$ is maximised -- this is precisely the
choice used in the E-step of Algorithm~\ref{algo:diametric}. With
binary values for $\beta_{ij}$, the M-Step~\eqref{eq:5} also
reduces to the version followed by Algorithm~\ref{algo:diametric}.

An alternative, perhaps better view is obtained by regarding
diametrical clustering as a special case of mixture-modelling
where a hard-assignment rule is used. Now, if all mixture
components have the \emph{same, positive} concentration parameter
$\kappa$, then while computing $\beta_{ij}$ via~(\ref{eq:31}) we
may ignore $\kappa$ altogether, which reduces
Algorithm~\ref{algo:mow} to Algorithm~\ref{algo:diametric}.

Given this interpretation of diametrical clustering, it is natural
to expect that the additional modelling power offered by mixtures
of Watson distributions might lead to better clustering. This is
indeed the case, as indicated by some of our experiments in
Section~\ref{sec:clustering} below, where we show that merely
including the concentration parameter $\kappa$ can lead to
improved clustering accuracies, or to clusters with higher quality
(in a sense that will be made more precise below).

\section{Experiments}
\label{sec:experiments} We now come to numerical results to assess
the methods presented. We divide our experiments into two groups.
The first group comprises numerical results that illustrate
accuracy of our approximation to $\kappa$. The second group
supports our claim that the extra modelling power offered by moWs
also translates into better clustering results.

\subsection{Estimating $\kappa$}
\label{sec:kappa.expt} We show two representative experiments to
illustrate the accuracy of our approximations. The first set
(\S\ref{sec:kappa.bij}) compares our approximation with that
of~\cite{avleen}, as given by~(\ref{eq:bijral}). This set
considers the Watson case, namely $a=1/2$ and varying
dimensionality $c=p/2$. The second set (\S\ref{sec:kappa.gen}) of
experiments shows a sampling of results for a few values of $c$
and $\kappa$ as the parameter $a$ is varied. This set illustrates
how well our approximations behave for the general nonlinear
equation~\eqref{eq:61}.

\subsubsection{Comparison with the BBG approximation for the Watson case}
\label{sec:kappa.bij} Here we fix $a=1/2$, and vary $c$ on an
exponentially spaced grid ranging from $c=10$ to $c=10^4$. For
each value of $c$, we generate geometrically spaced values of the
``true'' $\kappa_*$ in the range $[-200c, 200c]$. For each choice
of $\kappa_*$ picked within this range, we compute the ratio $r =
g(1/2, c, \kappa_*)$ (using~\textsc{Mathematica} for high
precision). Then, given $a=1/2$, $c$, and $r$, we estimate
$\kappa_*$ by solving $\kappa \approx g^{-1}(1/2,c,r)$ using
$BBG(r)$, $L(r)$, $B(r)$, and $U(r)$, given by  \eqref{eq:bijral},
\eqref{eq:lower}, \eqref{eq:mid}, and \eqref{eq:upper},
respectively.

\begin{figure}[!h]
  \centering
  \hskip-12pt\includegraphics[width=.35\linewidth]{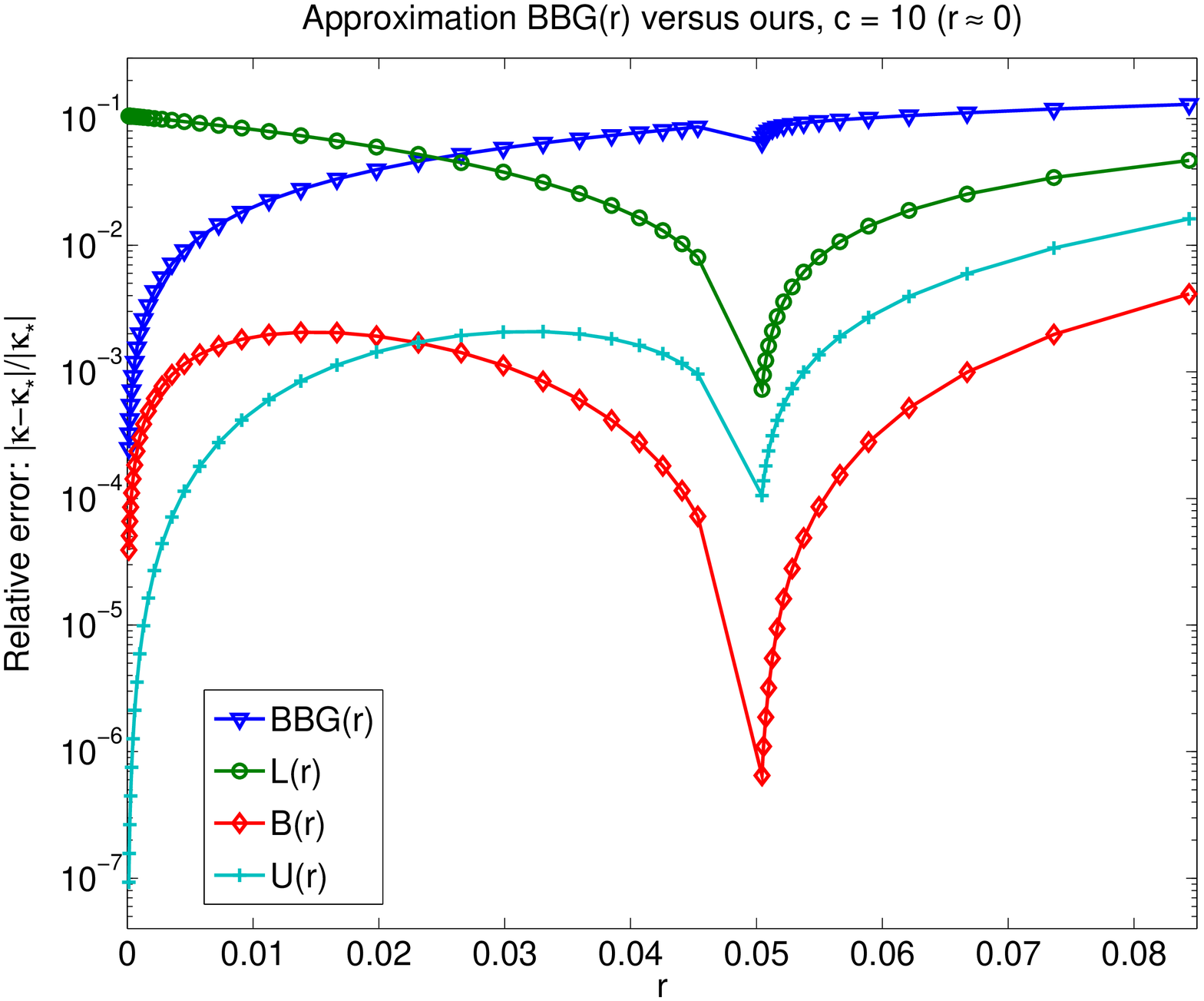}\hskip-8pt
  \includegraphics[width=.35\linewidth]{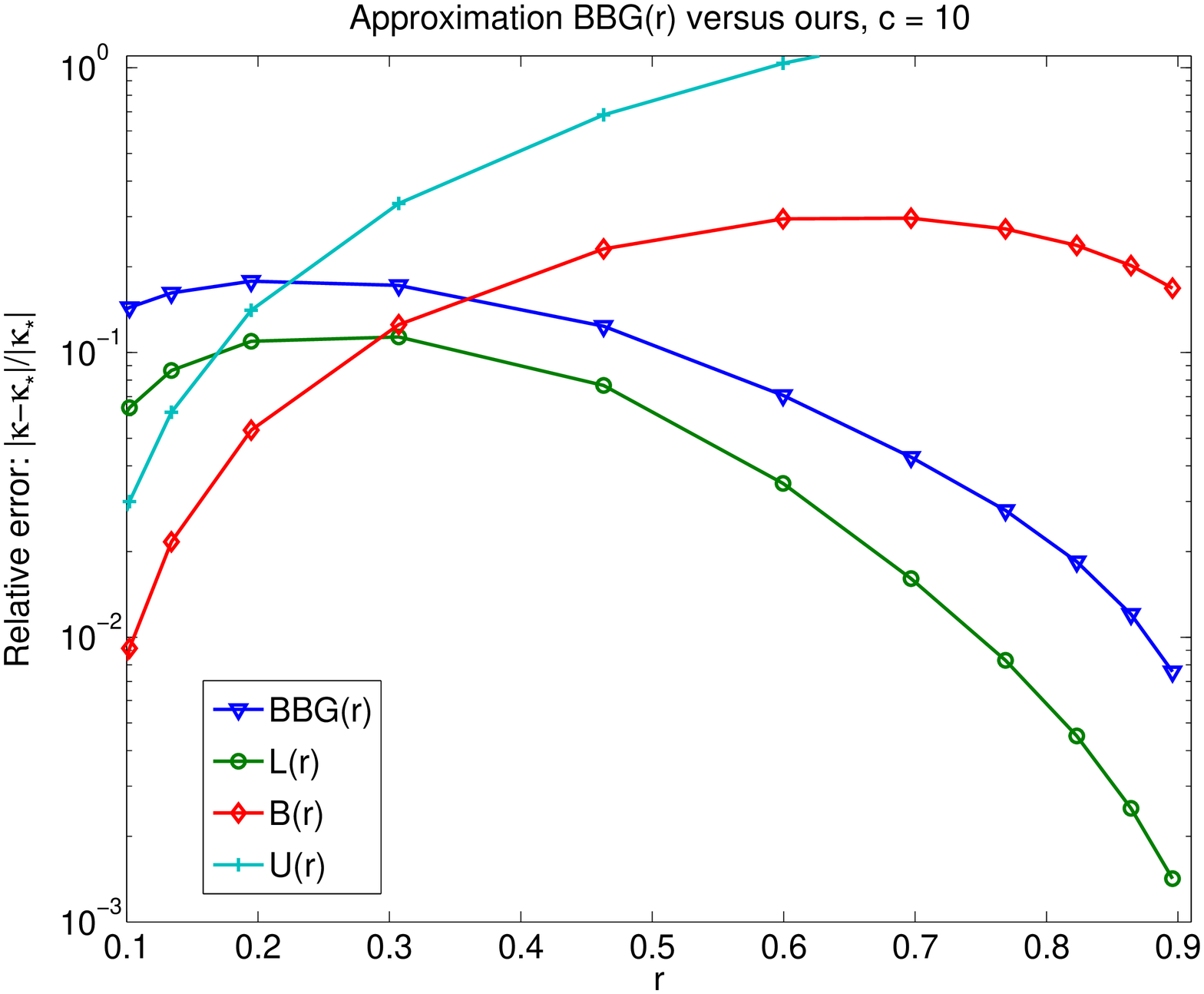}\hskip-8pt
  \includegraphics[width=.35\linewidth]{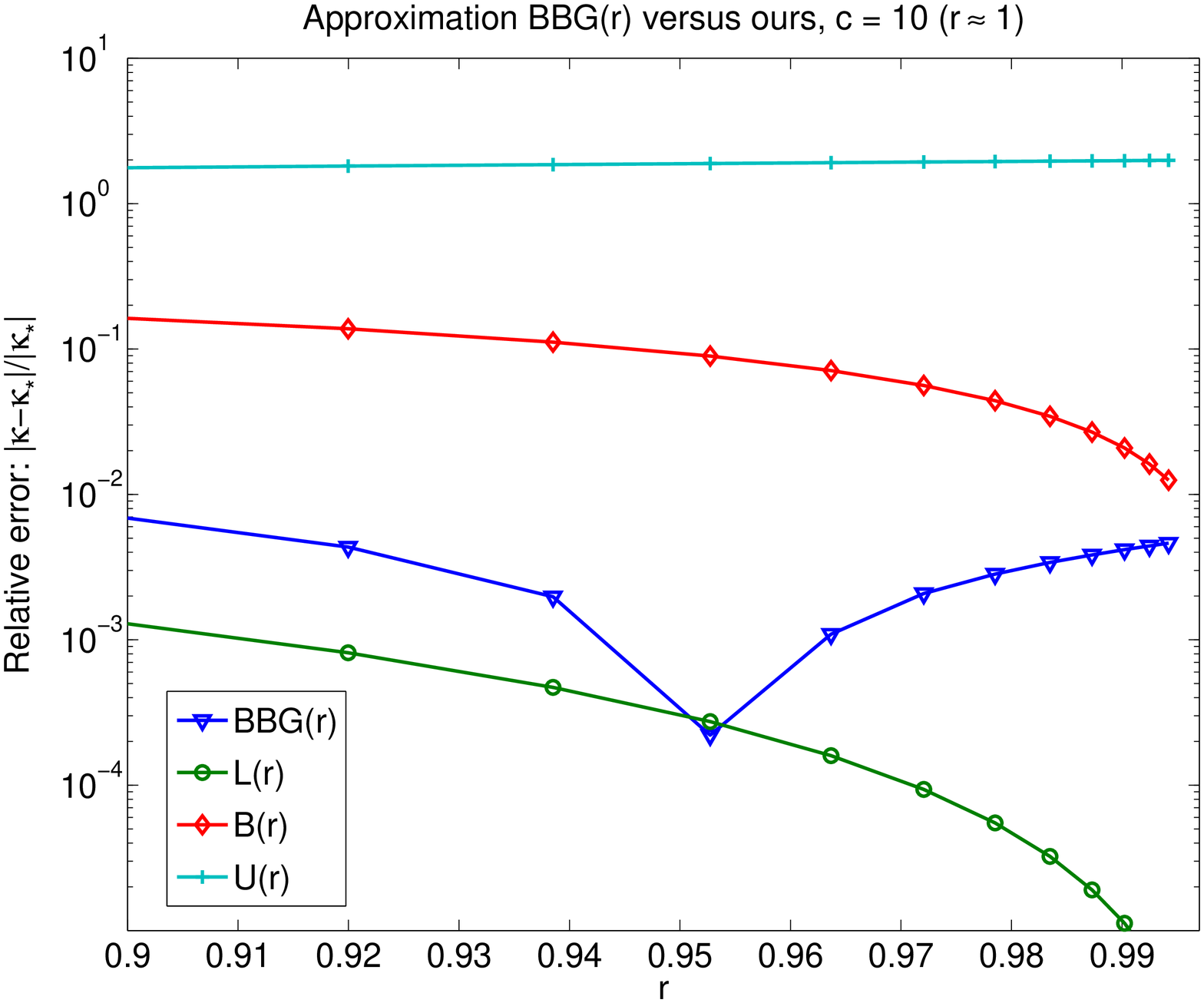}\\
  \hskip-12pt\includegraphics[width=.35\linewidth]{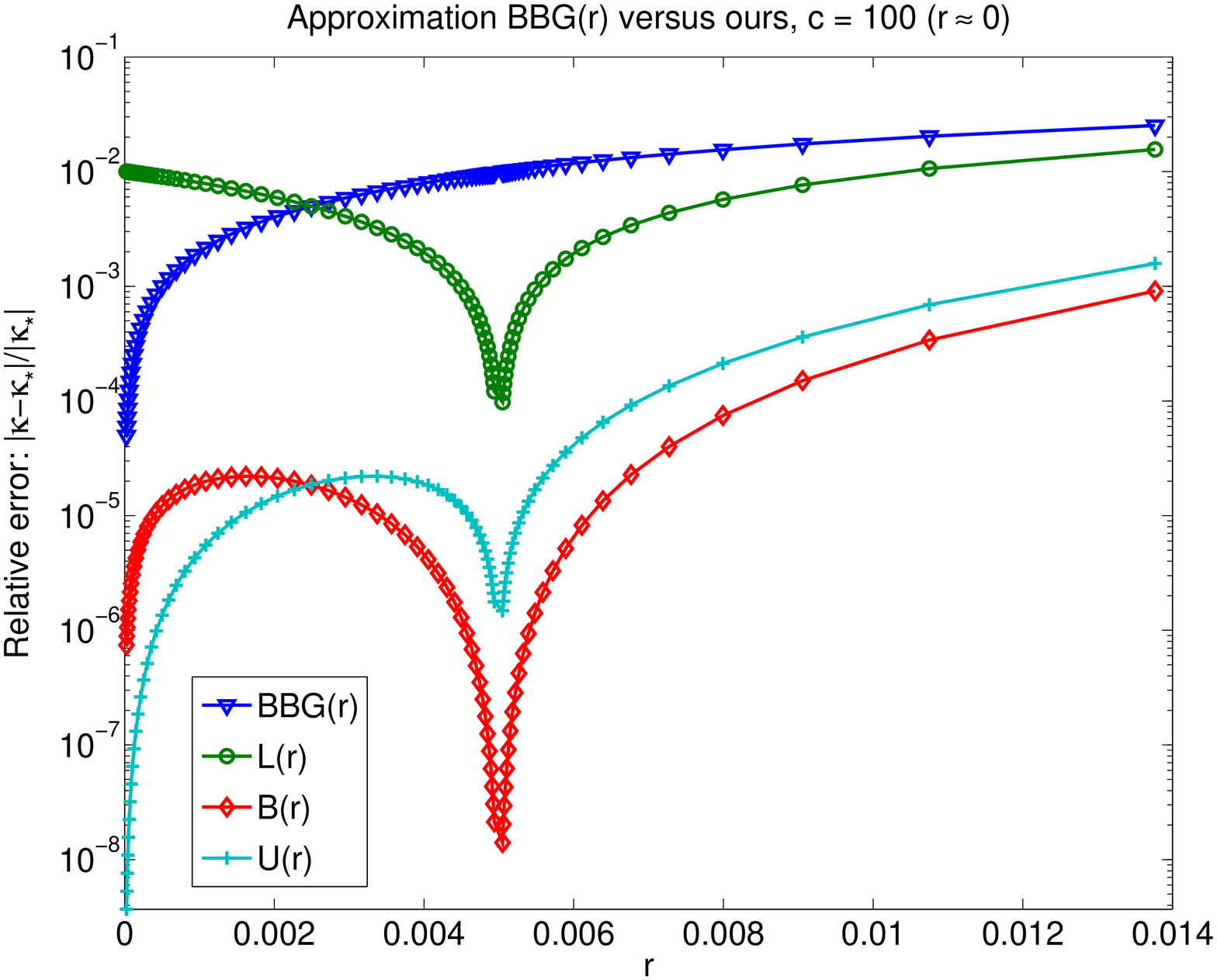}\hskip-8pt
  \includegraphics[width=.35\linewidth]{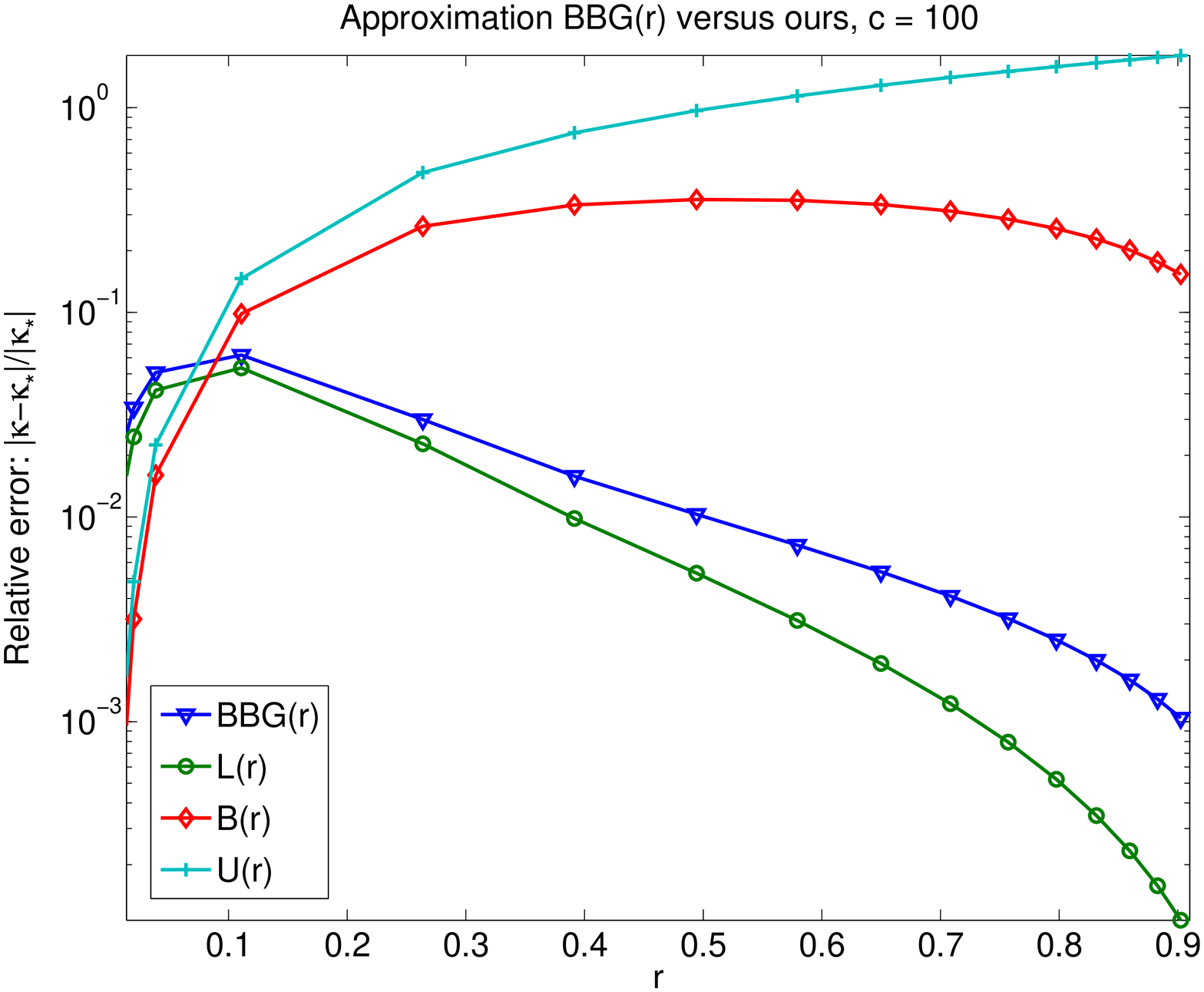}\hskip-8pt
  \includegraphics[width=.35\linewidth]{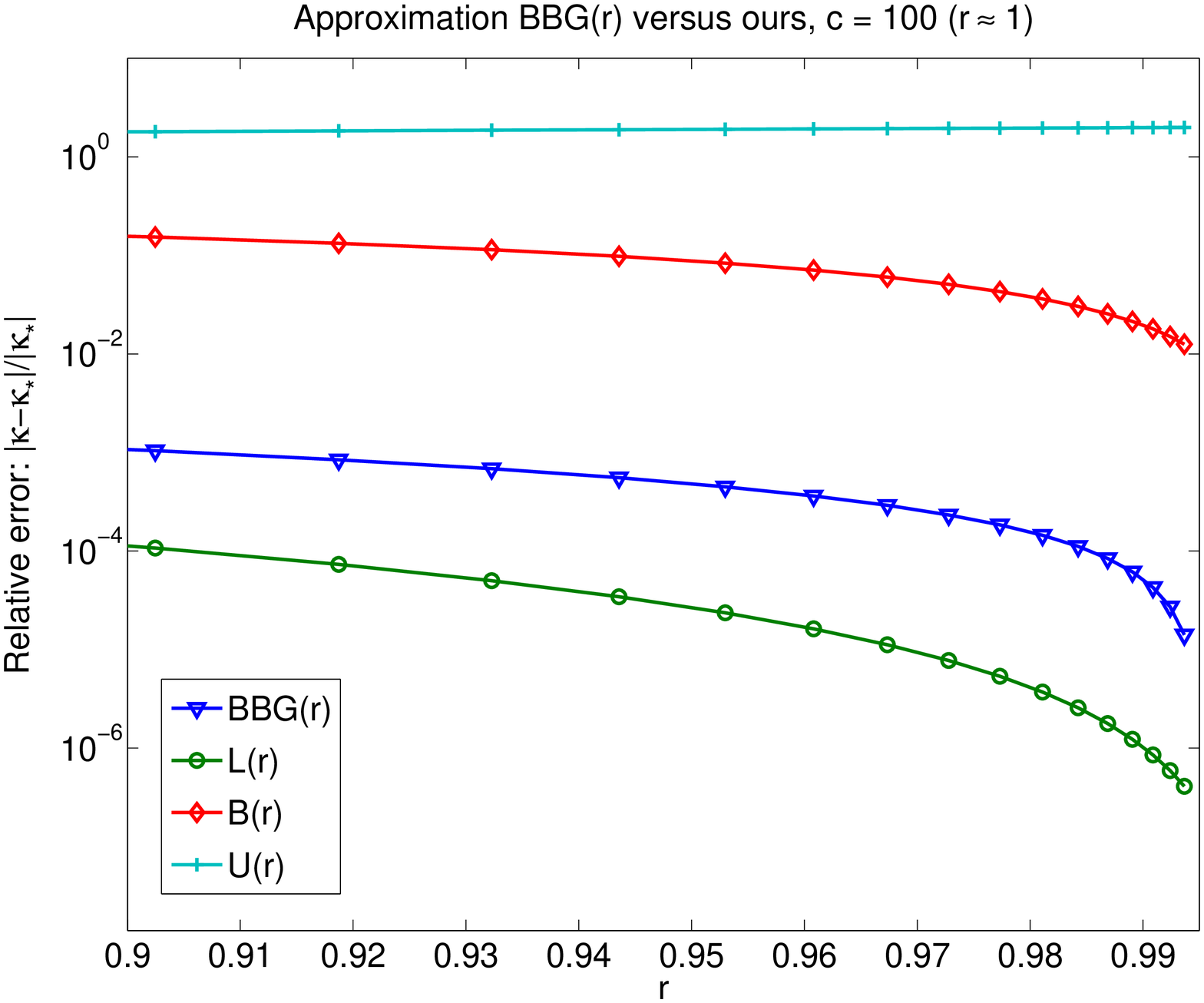}\\
  \hskip-12pt\includegraphics[width=.35\linewidth]{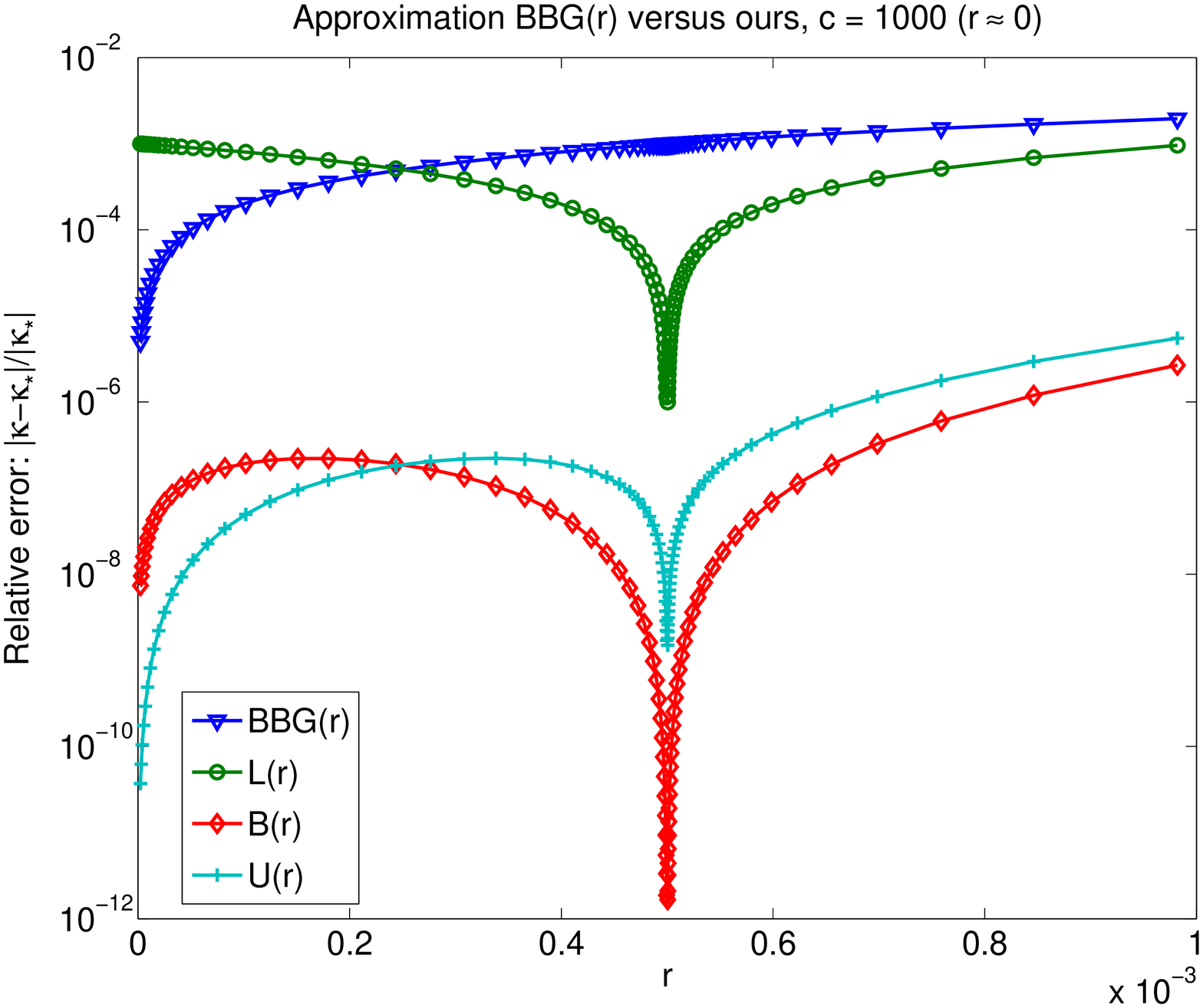}\hskip-8pt
  \includegraphics[width=.35\linewidth]{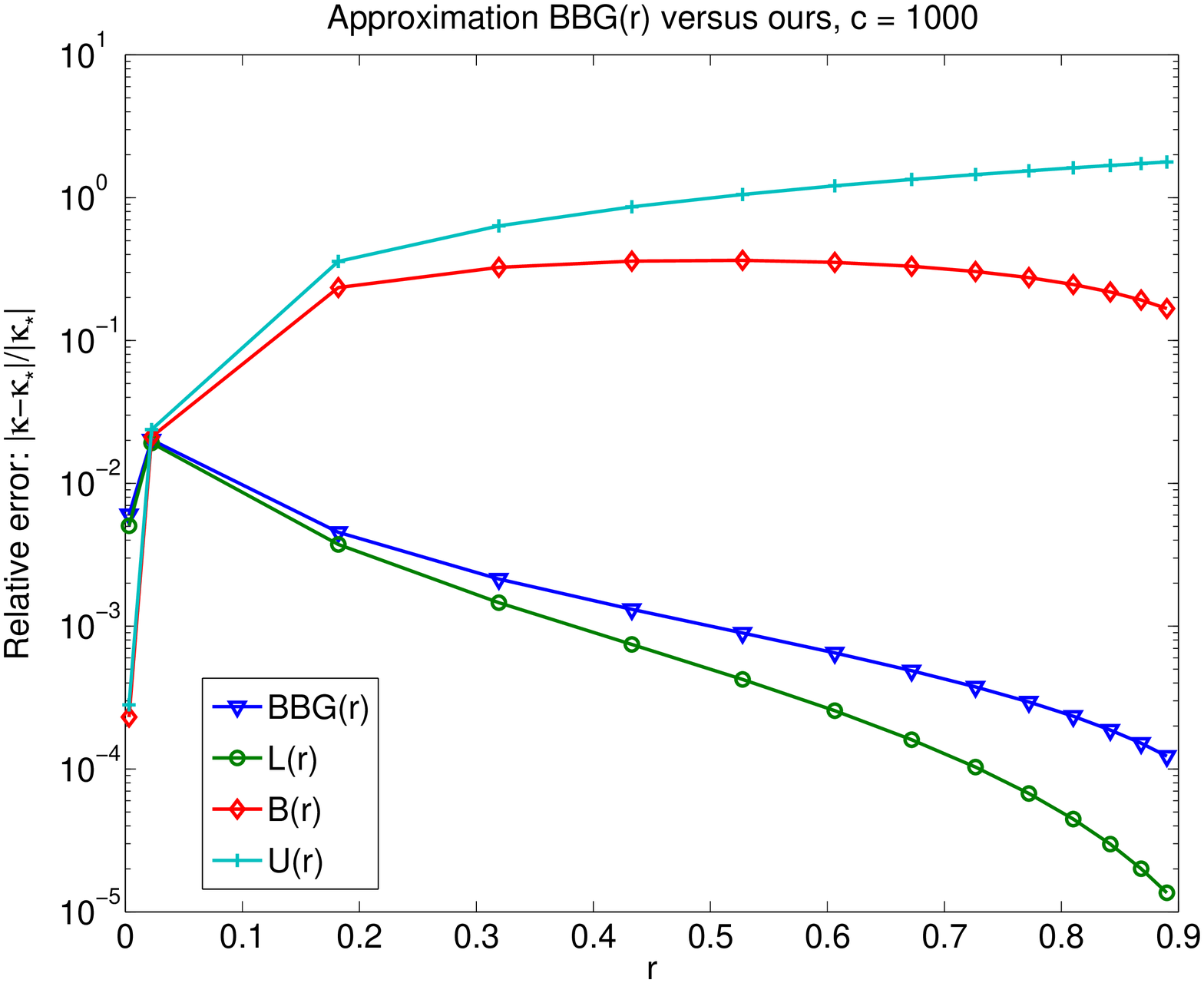}\hskip-8pt
  \includegraphics[width=.35\linewidth]{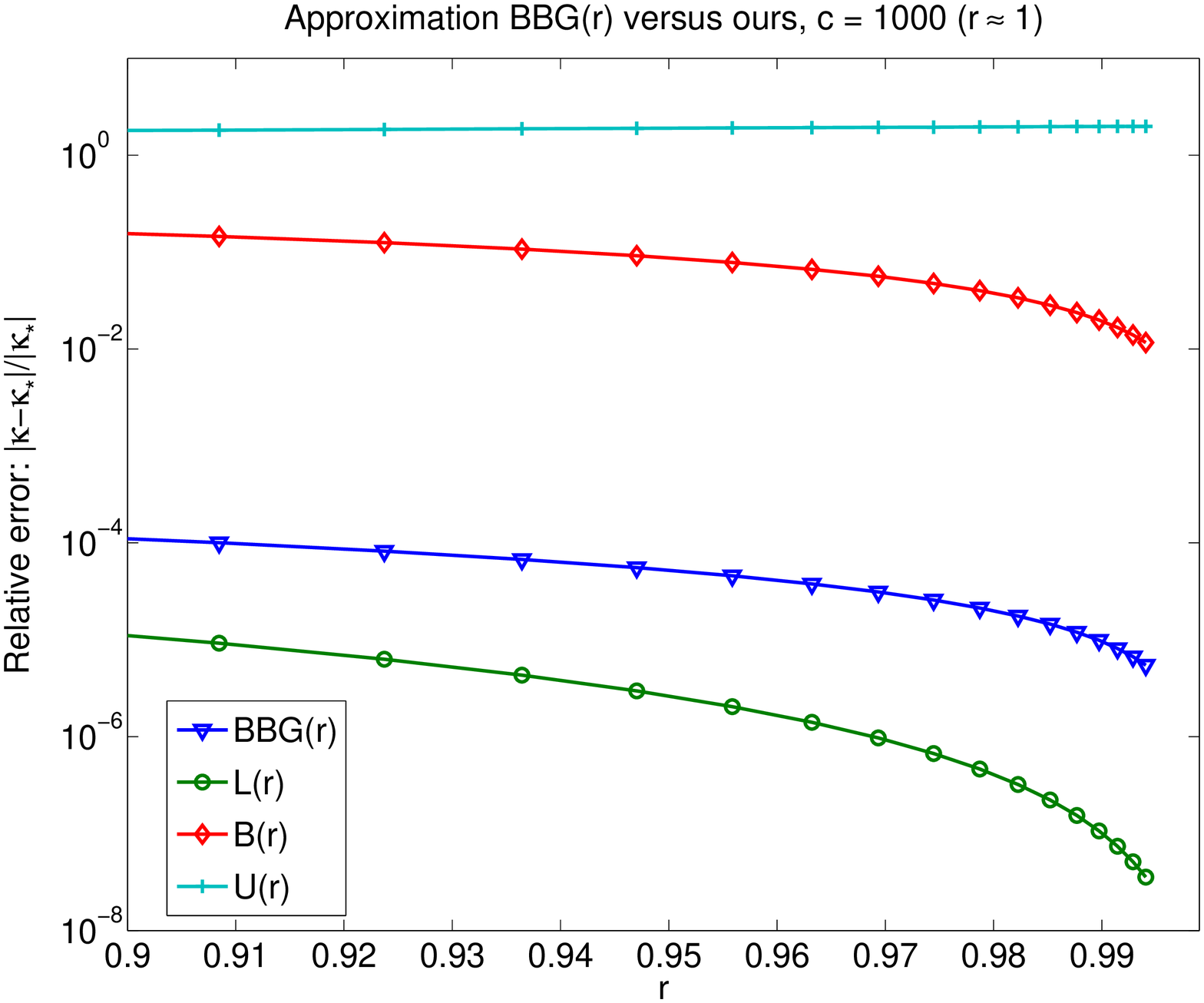}
  \caption{Relative errors $|\hat{\kappa}-\kappa_*|/|\kappa_*|$ of $BBG(r)$, $L(r)$, $B(r)$, and $U(r)$ for $c \in \set{10,100,1000,10000}$ as $r$ varies between $(0,1)$. The left column shows errors for ``small'' $r$ (i.e., $r$ close to $0$), the middle column shows errors for ``mid-range'' $r$, and the last column shows errors for the ``high'' range ($r \approx 1$).}
  \label{fig.one}
\end{figure}
Figure~\ref{fig.one} shows the results of computing these
approximations. Two points are immediate from the plots: (i)
approximation $L(r)$ is more accurate than $BBG(r)$ across almost
the whole range of dimensions and $r$ values; and (ii) for small
$r$, $BBG(r)$ can be more accurate than $L(r)$, but in this case
both $U(r)$ and $B(r)$ are much more accurate.

\subsubsection{Comparisons of the approximation for fixed $c$ and varying $a$}
\label{sec:kappa.gen} 
\begin{figure}[!htbp]
  \centering
  \includegraphics[width=.45\linewidth]{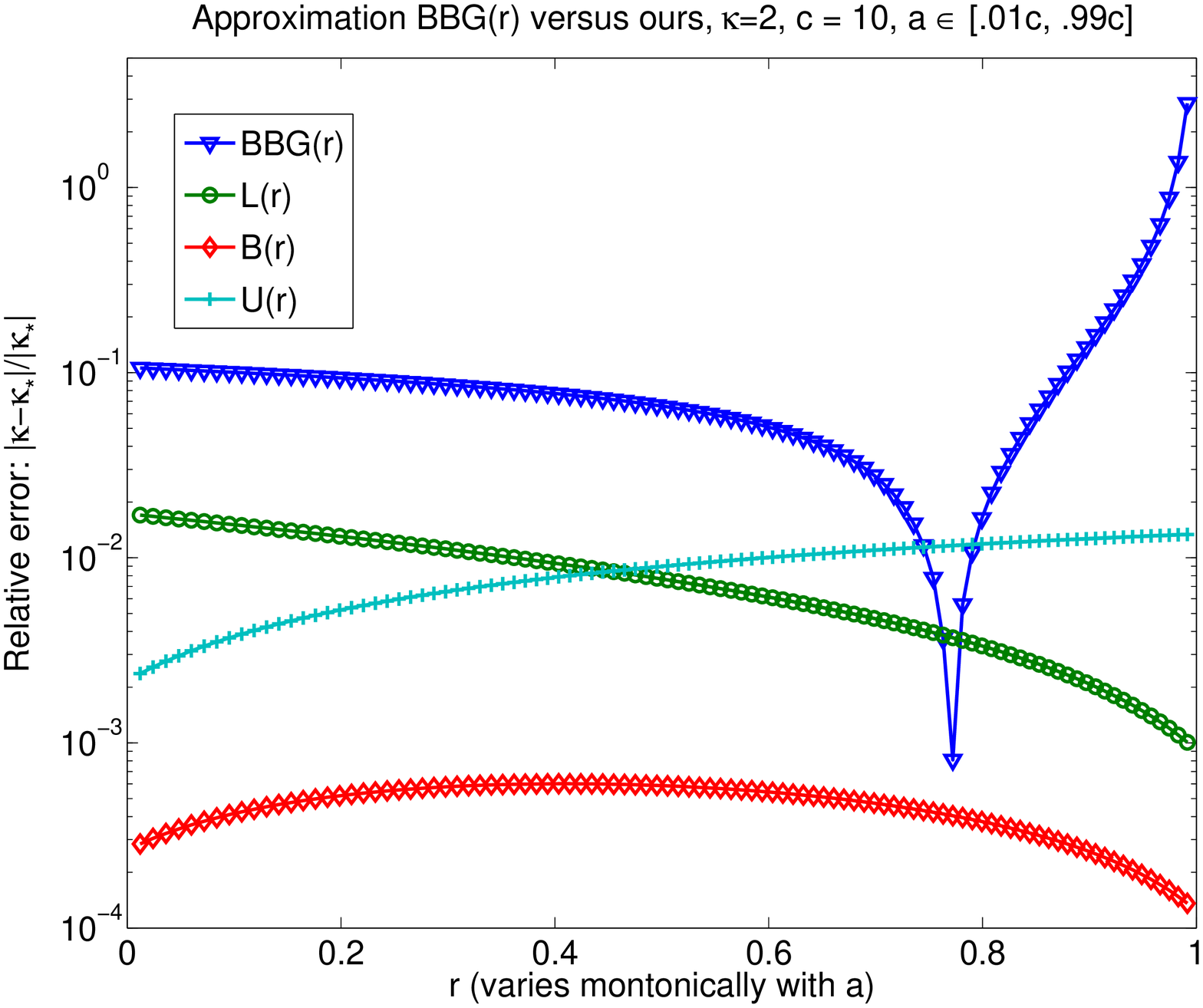}   \includegraphics[width=.45\linewidth]{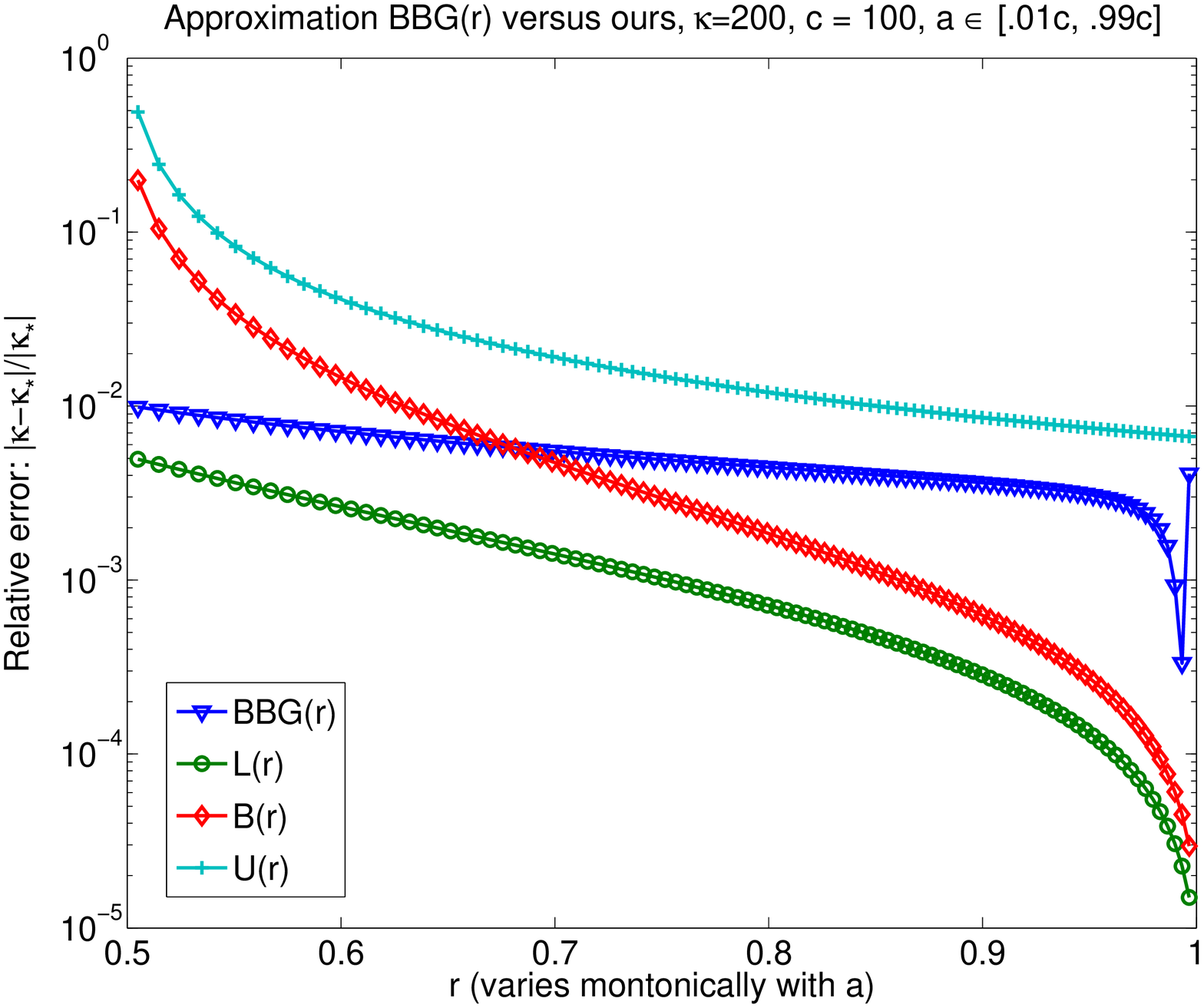}\\
  \includegraphics[width=.45\linewidth]{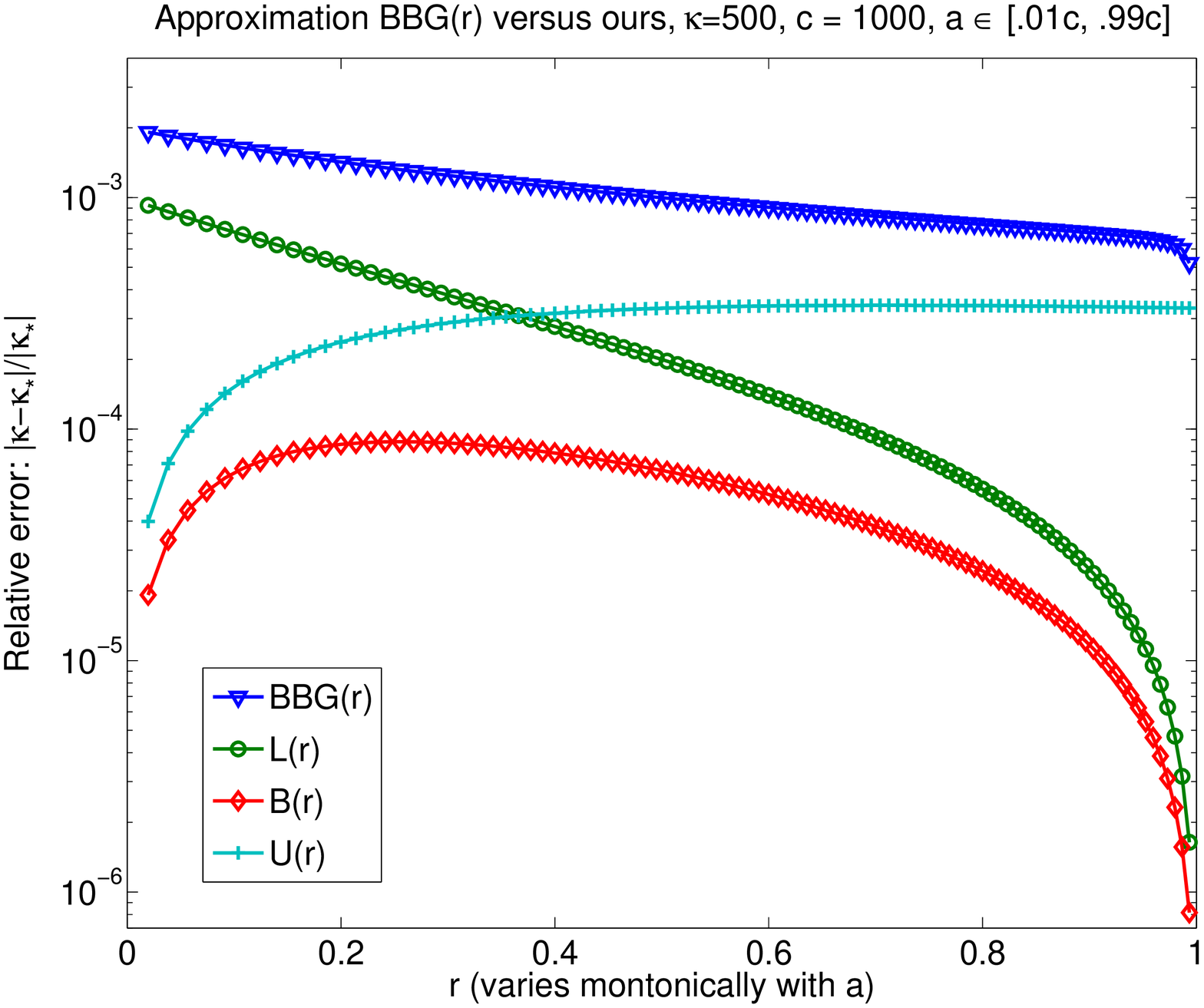}   \includegraphics[width=.45\linewidth]{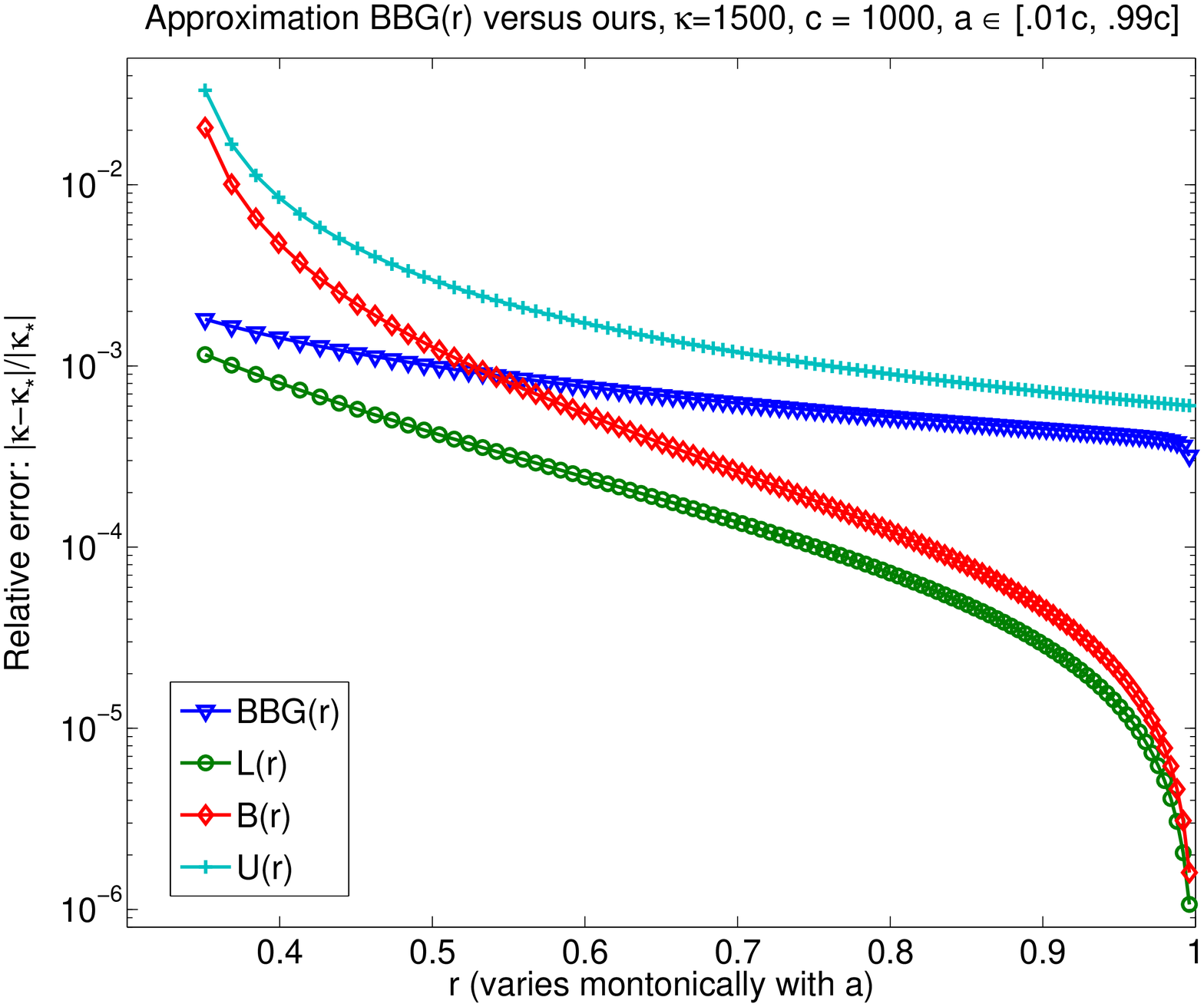}
  \caption{Relative errors of $BBG(r)$, $L(r)$, $B(r)$, and $U(r)$ for different sets of $c$ and $\kappa$ values, as $a$ is varied from $0.01c$ to $0.99c$.}
  \label{fig.kapp3}
\end{figure}
In our next set of experiments, we chose a
few values of $c$ and $\kappa$ (see Figure~\ref{fig.kapp3}), and
varied $a$ linearly to lie in the range $[0.01c, 0.99c]$.
Figure~\ref{fig.kapp3} reports the relative errors of
approximation incurred by the various approximations.

From the plots it is clear that one of $L(r)$, $B(r)$, or $U(r)$
always yields results more accurate than $BBG(r)$. The various
results suggest the following rough rule-of-thumb: prefer $U(r)$
for $0<r<a/(2c)$, prefer $B(r)$ for $a/(2c)\leq{r}<2a/\sqrt{c}$
and prefer $L(r)$ for $2a/\sqrt{c}\leq{r}<1$.

\subsection{Clustering using mW distributions}
\label{sec:clustering} Now we turn to our second set of
experiments. Below we show results of two experiments: (i) with
synthetic data, where a desired ``true-clustering'' is known; and
(ii) with gene expression data for which previously axially
symmetric clusters have been considered.

For both our experiments, we compare moW (Algorithm~\ref{algo:mow}
with~\eqref{eq:31} for the E-step) against the diametrical
clustering procedure of~\citet{diametric}. The key aim of the
experiments is to show that the extra modelling power offered by a
mixture of mW distributions can provide clustering results better
than plain diametrical clustering.

\subsubsection{Synthetic data}
We generated data that merely exhibit axial symmetry and have
varying degrees of concentration around given mean directions.
Since both the diametrical method as well as moW model axial
symmetry they can be fairly compared on this data. The distinction
comes, however, where moW further models concentration (via
$\kappa$), and in case the generated data is sufficiently
concentrated, this modelling translates into empirically superior
performance. Naturally, to avoid unfairly skewing results in
favour of moW, we do not compare it against diametrical clustering
on synthetic data sampled from a mixture of mW distributions as
moW explicitly optimises such a model.

For our data generation we need to sample points from
$W_p(\kappa,\m)$, for which we invoke a simplified version of the
powerful Gibbs sampler of~\citep{hoff} that can simulate
Bingham-von Mises-Fisher distributions. We note here that Bingham
distribution is parametrised by a matrix $\bm A$, and to use it
for sampling Watson distributions, we merely need to realise that
$\bm{A}=\kappa \m\m^T$.

With the sampling code in hand, we generate synthetic datasets
with varying concentration as follows. First, two random unit
vectors $\m_1, \m_2 \in \mathbb{P}^{29}$ are selected. Then, we
fix $\kappa_1=3$ and  sample $200$ points from
$W_3(\kappa_1,\m_1)$. Next, we vary $\kappa_2$ in the set $\set{3,
10, 20, 50, 100}$, and generate $200$ points for each value of
$\kappa_2$ by sampling from $W_p(\kappa_2,\m_2)$. Finally, by
mixing the $\kappa_1$ component with each of the five $\kappa_2$
components we obtain five datasets ${\cal X}_t$ ($1 \le t \le 5$).

Each of these five datasets is then clustered into two clusters,
using moW and diametrical clustering. Both algorithms are run ten
times each to smooth out the effect of random initializations.
Table~\ref{tab:synth} shows the results of clustering by
displaying the accuracy which measures the percentage of data
points that were assigned to the ``true'' clusters (i.e., the true
components in the mixture). The accuracies strongly indicate that
explicit modelling of concentration leads to better clustering as
$\kappa_2$ increases. In other words, larger $\kappa_2$ makes
points from the second cluster more concentrated around $\pm\m_2$,
thereby allowing easier separation between the clusters.

\begin{table}\small
  \centering
  \caption{Percentages of accurately clustered points for diametrical clustering vs.\ moW (over 10 runs). Since this is simulated data, we knew the cluster labels. The accuracy is then computed by matching the predicted labels with the known ones. In line with the theory, with increasing concentration the modelling power offered by moW shows a clear advantage over ordinary diametrical clustering.}\label{tab:synth}
  \begin{tabular}{c|c|c}
    $\kappa_2$ & Diametrical (avg/best/worst)-\%& moW (avg/best/worst)-\%\\
    \hline
    $3$ &  52.65 / 56.50 / 51.50 & 51.65 / 53.50 / 50.50\\
    $10$  & 52.75 / 56.00 / 50.50 & 54.10 / 57.00 / 50.00\\
    $20$ & 57.60 / 64.00 / 51.50 & 74.45 / 87.00 / 63.50\\
    $50$ & 66.00 / 78.50 / 50.00 & 99.50 / 99.50 / 99.50\\
    $100$ & 71.20 / 81.00 / 55.00 & 100.00 / 100.00 / 100.00\\
    \hline
  \end{tabular}
\end{table}

\subsubsection{Real Data}
\label{sec:realdata} We now compare clustering results of moW with
those of diametrical clustering on three gene microarray datasets
that were also used in the original diametrical clustering
paper~\citep{diametric}. These datasets are: (i) Human
Fibroblasts~\citep{iyer}; (ii) Yest Cell Cycle~\citep{spellman};
and (iii) Rosetta yeast~\citep{rosetta}. The respective matrix
sizes that we used were: (i) $517\times 12$; (ii) $696\times 82$;
and (iii) $900 \times 300$ (these 900 genes were randomly selected
from the original $5245$).

Since we do not have ground-truth clusterings for these datasets,
we validate our results using internal measures. Specifically, we
compute two scores: \emph{homogeneity} and \emph{separation},
which are defined below by $H_{\text{avg}}$ and $S_{\text{avg}}$,
respectively. Let ${\cal X}_j \subset {\cal X}$ denote cluster
$j$; then we define
\begin{align}
  \label{havg}
  H_{\text{avg}} & = \frac{1}{n}\sum_{j=1}^K \sum_{\x_i \in {\cal X}_j} (\x_i^T\m_j)^2,\\
  \label{savg}
  S_{\text{avg}} & = \frac{1}{\sum_{j \neq l}|{\cal X}_j||{\cal X}_{l}|}\sum_{j \neq l}|{\cal X}_j||{\cal X}_{l}|\min(\m_j^T\m_{l}, -\m_j^T\m_{l}).
\end{align}
We note a slight departure from the standard in our definitions
above. In~\eqref{havg}, instead of summing over $\x_i^T\m_j$, we
sum over their squares, while in~\eqref{savg}, instead of
$\m_j^T\m_{l}$, we use $\min(\m_j^T\m_{l},-\m_j^T\m_{l})$ because
for us $+\m_j$ and $-\m_j$ represent the same cluster.

We note that  diametrical clustering optimises precisely the
criterion~\eqref{havg}, and is thus favoured by our criterion.
Higher values of $H_{\text{avg}}$ mean that the clusters have
higher intra-cluster cohesiveness, and thus are ``better''
clusters. In contrast, lower values of $S_{\text{avg}}$ mean that
the inter-cluster dissimilarity is high, i.e., better separated
clusters.

\begin{table}\small
  \caption{Clustering accuracy on gene-expression datasets (over 10 runs). Noticeable differences (i.e., $> 0.02$) between the algorithms are highlighted in bold.}\label{tab:real}
  \centering
  \begin{tabular}{c|c|c}
    Method & Diametrical (avg/best/worst) & moW (avg/best/worst)\\
    \hline
    \textbf{Yeast-4}&\\
    Homogeneity & 0.38 / 0.38 / 0.38 & 0.37 / 0.37 / 0.37\\
    Separation & -0.00 / -0.23 / 0.24 & {\bf -0.04} / -0.23 / {\bf 0.20}\\
    \hline
    \textbf{Yeast-6} &\\
    Homogeneity & 0.41 / 0.41 / 0.40 & 0.41 / 0.41 / 0.40\\
    Separation  & -0.06 / -0.15 / 0.14 & -0.07 /{\bf -0.20} / 0.13\\
    \hline
    \textbf{Rosetta-2} &\\
    Homogeneity & 0.16 / 0.17 / 0.16 & 0.16 / 0.17 / 0.16\\
    Separation  & 0.24 / 0.08 / 0.28 & {\bf -0.20} / {\bf -0.28} / {\bf 0.09}\\
    \hline
    \textbf{Rosetta-4}&\\
    Homogeneity & 0.23 / 0.23 / 0.23 & 0.23 / 0.23 / 0.23\\
    Separation  & -0.01 / -0.08 / 0.16 & -0.03 / -0.09 / {\bf 0.12}\\
    \hline
    \textbf{Fibroblast-2} &\\
    Homogeneity & 0.70 / 0.70 / 0.70 & 0.70 / 0.70 / 0.70\\
    Separation  & 0.26 / -0.65 / 0.65 & {\bf -0.01} / -0.65 / 0.65\\
    \hline
    \textbf{Fibroblast-5} &\\
    Homogeneity & 0.78 / 0.78 / {\bf 0.78} & 0.76 / 0.76 / 0.75\\
    Separation  & -0.05 / -0.28 / 0.40 & {\bf -0.12} / -0.30 / {\bf 0.35}\\
    \hline
  \end{tabular}
\end{table}

Table~\ref{tab:real} shows results yielded by diametrical
clustering and moW on the three different gene datasets. For each
dataset, we show results for two values of $K$. The
$H_{\text{avg}}$ values indicate that moW yields clusters having
approximately the same intra-cluster cohesiveness as diametrical.
However, moW attains better inter-cluster separation as it more
frequently leads to lower $S_{\text{avg}}$ values.

\section{Conclusions}
We studied the multivariate Watson distribution, a fundamental
tool for modelling axially symmetric data. We solved the difficult
nonlinear equations that arise in maximum-likelihood parameter
estimation. In high-dimensions these equations pose severe
numerical challenges. We derived tight two-sided bounds that led
to approximate solutions to these equations; we also showed our
solutions to be accurate.  We applied our results to
mixture-modelling with Watson distributions and consequently
uncovered a connection to the diametrical clustering algorithm
of~\citep{diametric}. Our experiments showed that for clustering
axially symmetric data, the additional modelling power offered by
mixtures of Watson distributions can lead to better clustering.
Further refinements to the clustering procedure, as well as other
applications of Watson mixtures in high-dimensional settings is
left as a task for the future.

\subsection*{Acknowledgements}
The first author thanks Prateek Jain for initial discussions
related to Watson distributions. The second author acknowledges
support of the Russian Basic Research Fund (grant 11-01-00038-a).

\section*{References}
{\small \setlength{\bibsep}{3pt}

\begin{thebibliography}{22}
\expandafter\ifx\csname
natexlab\endcsname\relax\def\natexlab#1{#1}\fi
\expandafter\ifx\csname url\endcsname\relax
  \def\url#1{\texttt{#1}}\fi
\expandafter\ifx\csname
urlprefix\endcsname\relax\def\urlprefix{URL }\fi

\bibitem[{Andrews et~al.(1999)Andrews, Askey, and Roy}]{AAR}
Andrews, G.~E., Askey, R., Roy, R., 1999. Special functions.
Cambridge
  University Press.

\bibitem[{Banerjee et~al.(2003)Banerjee, Dhillon, Ghosh, and Sra}]{bdgs03:kdd}
Banerjee, A., Dhillon, I.~S., Ghosh, J., Sra, S., 2003. Generative
model-based
  clustering of directional data. In: Proceedings of The Ninth ACM SIGKDD
  International Conference on Knowledge Discovery and Data Mining({KDD}-2003).
  pp. 19--28.

\bibitem[{Banerjee et~al.(2005)Banerjee, Dhillon, Ghosh, and Sra}]{suv.vmf}
Banerjee, A., Dhillon, I.~S., Ghosh, J., Sra, S., Sep 2005.
{Clustering on the
  Unit Hypersphere using von Mises-Fisher Distributions}. J. Machine Learning
  Research 6, 1345--1382.

\bibitem[{Bijral et~al.(2007)Bijral, Breitenbach, and Grudic}]{avleen}
Bijral, A., Breitenbach, M., Grudic, G.~Z., 2007. {Mixture of
Watson
  Distributions: A Generative Model for Hyperspherical Embeddings}. In:
  Artificial Intelligence and Statistics (AISTATS 2007). pp. 35--42.

\bibitem[{Cuyt et~al.(2008)Cuyt, Petersen, Verdonk, Waadeland, and
  Jones}]{cuyt}
Cuyt, A., Petersen, V.~B., Verdonk, B., Waadeland, H., Jones,
W.~B., 2008.
  {Handbook of Continued Fractions for Special Functions}. Springer.


\bibitem[{Dempster et~al.(1977)Dempster, Laird, and Rubin}]{dlr77}
Dempster, A., Laird, N., Rubin, D., 1977. {Maximum Likelihood from
Incomplete
  Data Via the {EM} Algorithm}. {Journal of the Royal Statistical Society} 39.

\bibitem[{Dhillon et~al.(2003)Dhillon, Marcotte, and Roshan}]{diametric}
Dhillon, I.~S., Marcotte, E.~M., Roshan, U., 2003. Diametrical
clustering for
  identifying anti-correlated gene clusters. Bioinformatics 19~(13),
  1612--1619.

\bibitem[{Erd\'elyi et~al.(1953)Erd\'elyi, Magnus, Oberhettinger, and
  Tricomi}]{htf}
Erd\'elyi, A., Magnus, W., Oberhettinger, F., Tricomi, F.~G.,
1953. Higher transcendental functions. Vol.~1. McGraw Hill.

\bibitem[{Gautschi(1977)}]{gautschi77}
Gautschi, W., 1977. {Anomalous Convergence of a Continued Fraction
for Ratios  of Kummer Functions}. Mathematics of Computation
31~(140), 994--999.

\bibitem[{Gil et~al.(2007)Gil, Segura, and Temme}]{gil}
Gil, A., Segura, J., Temme, N.~M., 2007. {Numerical Methods for
Special  Functions}. Cambridge University Press.

\bibitem[{Gonzalez (1991)}]{Gonzalez} M.O.\,Gonzalez, Classical
Complex Analysis (Pure and Applied Mathematics), CRC Press, 1991.


\bibitem[{Hurwitz, Courant (1925)}]{HurwKur} A.\,Hurwitz,
R.\,Courant, Vorlesungen \"{u}ber allgemeine Fuktionentheorie und
Elliptische Fuktionen, Second Edition, Verlag von Julius Springer,
Berlin, 1925.


\bibitem[{Graham et~al.(1998)}]{gkp}
R.\,L.\,Graham, D.\,E.\,Knuth, and O.\,Patashnik, 1998. {Concrete
Mathematics}. Addison Wesley.

\bibitem[{Hoff(2009)}]{hoff}
Hoff, P.~D., 2009. {Simulation of the Matrix
Bingham{\textendash}von
  {Mises{\textendash}Fisher} Distribution, With Applications to Multivariate
  and Relational Data}. Journal of Computational and Graphical Statistics
  18~(2), 438--456.

\bibitem[{Hughes et~al.(2000)Hughes, Marton, Jones, Roberts, Stoughton, Armour,
  Bennett, Coffey, Dai, Shoemaker, Gachotte, Chakraburtty, Simon, Bard, and
  Friend}]{rosetta}
Hughes, T.~R., Marton, M.~J., Jones, A.~R., Roberts, C.~J.,
Stoughton, R.,
  Armour, C.~D., Bennett, H.~A., Coffey, E., Dai, H., Shoemaker, D.~D.,
  Gachotte, D., Chakraburtty, K., Simon, J., Bard, M., Friend, S.~H., 2000.
  Functional discovery via a compendium of expression profiles. Cell 102,
  109--126.

\bibitem[{Iyer et~al.(1999)Iyer, Eisen, Ross, Schuler, Moore, Lee, Trent,
  Staudt, Hudson, Boguski, Lashkari, Shalon, Botstein, and Brown}]{iyer}
Iyer, V.~R., Eisen, M.~B., Ross, D.~T., Schuler, G., Moore, T.,
Lee, J. C.~F.,
  Trent, J.~M., Staudt, L.~M., Hudson, J., Boguski, M.~S., Lashkari, D.,
  Shalon, D., Botstein, D., Brown, P.~O., 1999. {The Transcriptional Program in
  the Response of Human Fibroblasts to Serum}. Science 283~(5398), 83--87.

\bibitem[{Karp(2011)}]{Karp}
Karp, D., 2011. {Tur\'an's inequality for the Kummer function of
the phase
  shift of two parameters}. Journal of Mathematical Sciences 178~(2), 178--186.

\bibitem[{Karp and Sitnik(2010)}]{KS}
Karp, D., Sitnik, S.~M., 2010. Log-convexity and log-concavity of
  hypergeometric-like functions. Journal of Mathematical Analysis and
  Applications 364, 384--394.

\bibitem[{Mardia and Jupp(2000)}]{maju00}
Mardia, K.~V., Jupp, P., 2000. {Directional Statistics}, 2nd
Edition. John
  Wiley \& Sons.

\bibitem[{Spellman et~al.(1998)Spellman, Sherlock, Zhang, Iyer, Anders, Eisen,
  Brown, Botstein, and Futcher}]{spellman}
Spellman, P.~T., Sherlock, G., Zhang, M., Iyer, V.~R., Anders, K.,
Eisen, M.,
  Brown, P.~O., Botstein, D., Futcher, B., 1998. Comprehensive identification
  of cell cycle regulated gene of the yeast {S}accharomyces {C}erevisia by
  microarray hybridization. Mol. Bio. Cell 9, 3273--3297.

\bibitem[{Sra(2007)}]{suv.phd}
Sra, S., 2007. {Matrix Nearness Problems in Data Mining}. Ph.D.
thesis, Univ.
  of Texas at Austin.


\bibitem[{Tanabe et~al.(2007)Tanabe, Fukumizu, Oba, Takenouchi, and
  Ishii}]{tanabe}
Tanabe, A., Fukumizu, K., Oba, S., Takenouchi, T., Ishii, S.,
2007. {Parameter
  estimation for von Mises-Fisher distributions}. Computational Statistics
  22~(1), 145--157.

\bibitem[{Watson(1965)}]{watson65}
Watson, G.~S., 1965. {Equatorial distributions on a sphere}.
Biometrika
  52~(1-2), 193--201.

\end{thebibliography}

}

\appendix

\section{Mathematical Details}

\label{app:math}
This appendix includes mathematical details supporting the technical material
of the main text. While many of the facts are classic knowledge, some might be
found only in specialised literature. Thus, we have erred on the side of
including too much rather than too little.

\subsection{Hypergeometric functions}
\label{app:hyper}
Hypergeometric functions provide one of the richest classes of functions in
analysis. Indeed, any series with ratio of neighbouring terms equal to a rational function of the summation index is a constant multiple of the \emph{generalised hypergeometric} function $\hyperpq{p}{q}$ defined
by the power-series
\begin{equation}
  \label{eq:25}
  \hyperpq{p}{q}(a_1,\dots,a_p; c_1,\dots,c_q; z) = \sum_{k \geq 0}
   \frac{\risingF{a_1}{k}\cdots\risingF{a_p}{k}}
   {\risingF{c_1}{k}\cdots\risingF{c_q}{k}} \frac{z^k}{k!},
\end{equation}
where $\risingF{a}{k} = a (a+1)\dots(a+k-1)$ is the \emph{rising factorial} (often also denoted by the Pochhammer symbol $(a)_k$). Hypergeometric functions arise naturally
as solutions to certain differential equations; for a gentle introduction to
hypergeometric functions we refer the reader to~\citep{gkp}, while for a more
advanced treatment the reader may find~\citep{AAR} valuable.

In this paper, we restrict our attention to Kummer's confluent hypergeometric
function: $\hyperpq{1}{1}$, which is also denoted as $\kummer$. Moreover, we
limit our attention to the case of real valued arguments.

\vspace*{0pt}
\subsubsection{Some useful identities for $\kummer(a,c,x)$}
We list below some identities for $M$ that we will need for our analysis. To
ease the notational burden, we also use the shorthand $M_i \equiv M(a+i, c+i,
x)$; e.g.,~ $M_0\equiv M(a,c,x)$.
\begin{equation}
  \label{eq:deriv}
  \frac{d^n}{dx^n}M_0 = \frac{\risingF{a}{n}}{\risingF{c}{n}} M_n.
\end{equation}
\begin{equation}
  \label{eq:56}
  M_1 = \frac{c(1-c+x)}{ax}M_0 + \frac{c(c-1)}{ax}M_{-1}.
\end{equation}
\begin{equation}
  \label{eq:4}
  (c-a)M(a+1,c+2,x)=(c+1)M_1 - (a+1)M_2
\end{equation}
\begin{equation}
  \label{eq:11}
  (c-a)x M(a+2,c+3,x) = (c+1)(c+2)[M_2- M_1];
\end{equation}
\begin{equation}
  \label{eq:12}
  (a+1)xM_2 = (c+1)(x-c)M_1 + c(c+1)M_0
\end{equation}
\begin{equation}
  \label{eq:8}
  xM(a+2,c+3,x)= (c+2)[M_2-M(a+1,c+2,x)],
\end{equation}
Identity~\eqref{eq:deriv} follows inductively;~\eqref{eq:56} from~\citep[16.1.9c]{cuyt}; ~\eqref{eq:4} from~\citep[formula 6.4(4)]{htf}; ~\eqref{eq:11} on combining~\citep[formula 6.4(5)]{htf} with~\citep[formula 6.4(4)]{htf}; \eqref{eq:12} from~\eqref{eq:56} by replacing $c \to c+1$, $a\to a+1$;
and \eqref{eq:8} from~\citep[formula 6.4(5)]{htf}.

Now we build on the above identities to introduce a technical but crucial lemma.
\begin{lemma}
  \label{lemm:recurrence}
  The following identity holds for the Kummer function:
  \begin{multline}
   \label{eq:1}
   M_1^2 - M_2M_0 =\frac{(c-a)x}{c+1}\biggl[\frac{1}{c+1}M(a+1,c+2,x)^2-\frac{1}{c+2}M(a+2,c+3,x)M(a,c+1,x)\\
   + \frac{1}{c(c+1)}M(a+1,c+2,x)M(a+2,c+2,x)\biggr].
  \end{multline}
\end{lemma}
\begin{proof}
  Application of~\eqref{eq:4} and~\eqref{eq:11} yields after collecting terms
  \begin{multline}
    \frac{1}{c+1}M(a+1,c+2,x)^2 - \frac{1}{c+2}M(a+2,c+3,x)M(a,c+1,x)
    + \frac{1}{c(c+1)}M(a+1,c+2,x)M_2\\=\frac{a(a+1)}{c(c-a)^2}M_2^2
    -\frac{c(c+1)}{(c-a)^2x}M_2M_0 - \frac{(c+1)(a-x)}{(c-a)^2x}M_1^2\\
    + \frac{c(c+1)}{(c-a)^2x}M_1M_0 + \frac{ac(c+1) -x(a+c+2ac)}{(c-a)^2cx}M_2M_1.
  \end{multline}
  Application of this formula allows us to write the difference between the left-hand and right-hand sides of~\eqref{eq:1} as
  \begin{equation*}
    \text{lhs}-\text{rhs} =
    \frac{cM_1-aM_2}{c(c+1)(c-a)}\bigl(x(a+1)M_2 - (c+1)(x-c)M_1 -
    c(c+1)M_0\bigr) = 0,
  \end{equation*}
  where the equality to $0$ is due to~\eqref{eq:12}.~~$\square$
\end{proof}

\subsection{The Kummer ratio}

The central object of study in this paper is the
\emph{Kummer-ratio}:
\begin{equation}
  \label{eq:52}
  g(x)=g(a, c; x) := \frac{M'(a,c,x)}{M(a,c,x)} = \frac{a}{c}\frac{M(a+1,c+1,x)}{M(a,c,x)}.
\end{equation}
This ratio satisfies many fascinating properties; but of
necessity, we must content ourselves with only the essential
properties. In particular, our analysis focuses on the following:
(i) proving that $g$ is monotonic, and thereby invertible; and
(ii) obtaining bounds on the root of $g(x) - r = 0$. In the
sequel, it will be useful to use the slightly more general
function
\begin{equation} \label{eq.1}
 f_\mu(x) := \frac{M(a+\mu, c+\mu, x)}{M(a,c,x)},\qquad \mu>0,
\end{equation}
so that $g(x) = (a/c)f_1(x)$. Before proving monotonicity of $g$,
we derive two useful lemmas.
\begin{lemma}[Log-convexity]
  \label{lemm:logcvx}
  Let $c > a > 0$ and $x \ge 0$. Then the function
  \begin{align*}
    \mu\mapsto\frac{\Gamma(a+\mu)}{\Gamma(c+\mu)}\kummer(a+\mu, c+\mu, x)
    &=
    \sum\limits_{k=0}^{\infty}\frac{\Gamma(a+\mu+k)}{\Gamma(c+\mu+k)}\frac{x^k}{k!}
    =: h_{a,c}(\mu;x)\\
  \end{align*}
  is strictly log-convex on $[0,\infty)$ (note that $h$ is a function of
  $\mu$).
\end{lemma}
\begin{proof}
  Write the power-series expansion in $x$ for $h_{a,c}(\mu; x)$ as
  \[
  h_{a,c}(\mu;x)=\sum\limits_{k=0}^{\infty}h_k(a,c,\mu)\frac{x^k}{k!},~~~h_k(a,c,\mu)=\frac{\Gamma(a+\mu+k)}{\Gamma(c+\mu+k)}.
  \]
  Since log-convexity is additive it is sufficient to
  prove that $\mu \mapsto {h_k(a,c,\mu)}$ is log-convex. For this we
  compute the second-derivative
  \[
  \frac{\partial^2}{\partial\mu^2}\log{h_k(a,c,\mu)}=\psi'(a+\mu+k)-\psi'(c+\mu+k),
  \]
  where $\psi$ is the logarithmic derivative of the gamma function.  We
  need to show that this expression is positive when $c>a>0$,
  $k\geq{0}$ and $\mu\geq{0}$.  According to the Gauss formula~\citep[Theorem~1.6.1]{AAR}
  \[
  \psi(x)=\int\limits_{0}^{\infty}\left(\frac{e^{-t}}{t}-\frac{e^{-tx}}{1-e^{-t}}\right)dt,
  \]
  so that
  \[
  \psi''(x)=-\int\limits_{0}^{\infty}\frac{t^2e^{-tx}}{1-e^{-t}}dt<0.
  \]
  Hence the function $\psi'(x)$ is decreasing and our claim follows.\hfill$\square$
\end{proof}
\begin{lemma}
  \label{lemm:lcvx2}
  Let $c > a > 0$, and $x \ge 0$. Then the function
  \begin{equation*}
   \mu\mapsto\frac{\Gamma(a+\mu)}{\Gamma(c+\mu)}\kummer(c-a, c+\mu, x)=:\hat{h}_{a,c}(\mu;x)
  \end{equation*}
  is strictly log-convex on $[0,\infty)$.
\end{lemma}
\begin{proof}
  Using precisely the same argument as in the proof of
  Lemma~\ref{lemm:logcvx} we see that $\mu\mapsto\Gamma(a+\mu)/\Gamma(c+\mu)$
  is log-convex.  Next, the log-convexity of $\mu\mapsto \kummer(c-a;c+\mu;x)$
  has been proved by many authors (see, for instance,  \cite{KS} and
  references therein). Thus multiplicativity of log-convexity completes the
  proof.\hfill$\square$
\end{proof}
With these two lemmas in hand we are ready to prove the first main theorem.
\begin{theorem}[Monotonicity]
  \label{thm:monotone}
  Let $c > a > 0$. The function $x \mapsto f_\mu(x)$ is monotone increasing on
  $(-\infty,\infty)$, with $f_\mu(-\infty)=0$ and $f_\mu(\infty) = \Gamma(c+\mu)\Gamma(a)/\bigl(\Gamma(c)\Gamma(a+\mu)\bigr)$.
\end{theorem}
\begin{proof}
  We divide the proof into two cases: (i) $x \ge 0$, and (ii) $x < 0$.

  \emph{Case i:} Let $x \ge 0$. It follows from~(\ref{eq:deriv}) that
  \begin{equation*}
    \frac{d}{dx}M(a,c,x) = \frac{a}{c}\kummer(a+1,c+1,x),
  \end{equation*}
  whereby (using our compact notation) we have
  \begin{equation*}
    M_0^2 \fmu'(x) = \frac{a+\mu}{c+\mu}M_{\mu+1}M_0 - \frac{a}{c}M_\mu M_1.
  \end{equation*}
  We need to show that the above expression is positive, which amounts to showing
  \begin{equation}
    \label{eq:13}
    \frac{a+\mu}{c+\mu}\frac{M_{\mu+1}}{M_\mu} > \frac{a}{c}\frac{M_1}{M_0},
  \end{equation}
  or equivalently
  \begin{equation*}
    \frac{[\Gamma(a+\mu+1)/\Gamma(c+\mu+1)]M_{\mu+1}}{[\Gamma(a+\mu)/\Gamma(c+\mu)]M_\mu}
    >
    \frac{[\Gamma(a+1)/\Gamma(c+1)]M_1}{[\Gamma(a)/\Gamma(c)]M_0}.
  \end{equation*}
  The last inequality follows from Lemma~\ref{lemm:logcvx}. To see how, recall
  that if $\mu \mapsto h(\mu)$ is log-convex, then the function $\mu \mapsto
  h(\mu+\delta)/h(\mu)$ is increasing\footnote{Easily verified by noting that
    when $h$ is log-convex, its logarithmic derivative $h'(\mu)/h(\mu)$ is increasing,
    which immediately implies that the derivative of $h(\mu+\delta)/h(\mu)$ is
    positive.} for each fixed $\delta > 0$. Thus, in particular applying this
  property to $h_{a,c}(\mu; x)$ with $\delta = 1$ we have
  \begin{equation*}
    \frac{h_{a,c}(\mu+1; x)}{h_{a,c}(\mu; x)} > \frac{h_{a,c}(1;x)}{h_{a,c}(0;x)},
  \end{equation*}
  which is precisely the required inequality. This establishes the
  monotonicity. The value of $\fmu(\infty)$ follows from the
  asymptotic formula~\citep[Corollary~4.2.3]{AAR}:
  \begin{equation}\label{eq:kummer.asym}
    M(a,c,x) \sim  \frac{\Gamma(c)}{\Gamma(a)}\frac{e^x}{x^{c-a}}{_2F_0}(c-a,1-a;-;1/x),~~x\to\infty.
  \end{equation}

  \emph{Case ii:} Let $x < 0$. Like in Case~(i) we need to show that
  \begin{equation}
    \label{eq:rr1}
    \frac{[\Gamma(a+\mu+1)/\Gamma(c+\mu+1)]M_{\mu+1}}{[\Gamma(a+\mu)/\Gamma(c+\mu)]M_\mu} >
    \frac{[\Gamma(a+1)/\Gamma(c+1)]M_1}{[\Gamma(a)/\Gamma(c)]M_0}.
\end{equation}
but this time for $x<0$. Apply the Kummer transformation $\kummer(a;c;x)=e^{x}M(c-a;c;-x)$ and write $y=-x>0$ to get
\[
\frac{[\Gamma(a+\mu+1)/\Gamma(c+\mu+1)]M(c-a;c+\mu+1;y)}{[\Gamma(a+\mu)/\Gamma(c+\mu)]M(c-a;c+\mu;y)}>
\frac{[\Gamma(a+1)/\Gamma(c+1)]M(c-a;c+1;y)}{[\Gamma(a)/\Gamma(c)]M(c-a;c;y)}.
\]
Using the notation introduced in Lemma~\ref{lemm:lcvx2} the last inequality becomes
\[
\frac{\hat{h}_{a,c}(\mu+1;x)}{\hat{h}_{a,c}(\mu;x)}>\frac{\hat{h}_{a,c}(1;x)}{\hat{h}_{a,c}(0;x)},
\]
which holds as a consequence of the log-convexity of $\mu\mapsto\hat{h}_{a,c}(\mu;x)$.  Finally from the Kummer transformation
and formula~(\ref{eq:kummer.asym}) we have
\begin{equation}\label{eq:fmuneg.asymp}
f_{\mu}(x)=\frac{M(c-a;c+\mu;-x)}{M(c-a;c;-x)}\sim
\frac{\Gamma(c+\mu)}{\Gamma(c)}\frac{1}{(-x)^{\mu}}
\frac{{_2F_0}(a+\mu,1+a-c;-;-1/x)}{{_2F_0}(a,1+a-c;-;-1/x)}\to{0}~\text{as}~x\to-\infty.~\square
\end{equation}
\end{proof}

\subsection{Bounds on $\kappa$}
\label{app:bounds}
We now derive bounds on $\kappa(r)$, the solution to the equation $g(a,
c;\kappa)=r$. For ease of exposition we divide the bounds into two cases: (i) positive $\kappa(r)$ and (ii) negative $\kappa(r)$.

\begin{theorem}[Positive $\kappa$]\label{thm.poskappa}
  \label{thm:poskappa}
  Let $\kappa(r)$ be the solution to~(\ref{eq:61}); $c > a > 0$, and $r \in (a/c,1)$. Then, we have the bounds
  \begin{equation}
    \label{eq:xboundspos}
    \frac{rc-a}{r(1-r)}\left(1+ \frac{1-r}{c-a}\right) < \kappa(r) <
    \frac{rc-a}{2r(1-r)}\left(1 + \sqrt{1+
    \frac{4(c+1)r(1-r)}{a(c-a)}}\right)<\frac{rc-a}{r(1-r)}\left(1+\frac{r}{a}\right).
  \end{equation}
\end{theorem}
\begin{proof}
  \emph{Lower-bound.} To simplify notation we use $x=\kappa(r)$ below. Denote
  $r_1 = g(a+1,c+1; x)$. Now, replace $a \gets a+1$, $c \gets c+1$ and
  divide by $M_1$ in identity~(\ref{eq:56}) to obtain
  \begin{equation}
    \label{eq:2x}
    x = \frac{cr-a}{r(1-r_1)},
  \end{equation}
  where as before $r=g(a,c,x)$. To prove the lower bound in~\eqref{eq:xboundspos} we
  need to show that
  \begin{equation*}
    \frac{cr-a}{r(1-r_1)} > \frac{cr-a}{r(1-r)} + \frac{cr-a}{r(c-a)}.
  \end{equation*}
  Elementary calculation reveals that this inequality is equivalent to
  \begin{equation*}
    \frac{(c-a-1)r+1}{c-a+1-r} < r_1,
  \end{equation*}
  once we account for $cr-a > 0$ by our hypothesis. Plugging in the
  definitions of $r$ and $r_1$ we get:
  \begin{equation*}
    \frac{(c-a-1)aM_1 + cM_0}{(c-a+1)cM_0 - aM_1} < \frac{(a+1)M_2}{(c+1)M_1},
  \end{equation*}
  where as before we use $M_i = M(a+i,c+i,x)$. Cross-multiplying, we obtain
  \begin{equation*}
    h(x) := c(c-a+1)(a+1)M_2M_0 - (c+1)(c-a-1)aM_1^2 - c(c+1)M_1M_0 -
    a(a+1)M_2M_1 > 0.
  \end{equation*}
  Now on noticing that $c(c-a+1)(a+1)=ac(c-a)+c(c+1)$ and $(c-a-1)(c+1)a = ac(c-a)-a(a+1)$, we can regroup $h(x)$ to get
  \begin{equation*}
    h(x) = ac(c-a)[M_2M_0-M_1^2] + (M_2-M_1)[c(c+1)M_0-a(a+1)M_1].
  \end{equation*}
  Now identity~(\ref{eq:11}) yields
  \begin{equation*}
    M_2-M_1 = \frac{x(c-a)}{(c+1)(c+2)}M(a+2,c+3,x),
  \end{equation*}
  and a simple calculation shows that
  \begin{equation*}
  c(c+1)M_0-a(a+1)M_1 = \frac{ax(c-a)}{c+1}M(a+1,c+2,x) + (c-a)(c+a+1)M(a,c+1,x).
  \end{equation*}
  Substituting these formulae into $h(x)$ and using identity~(\ref{eq:1}) we
  get a rather complicated term
  \begin{multline}
    h(x) = -\frac{ac(c-a)^2x}{c+1}\biggl[ \frac{1}{c+1}M(a+1,c+2,x)^2 - \frac{1}{c+2}M(a+2,c+3,x)M(a,c+1,x)\\
    + \frac{1}{c(c+1)}M(a+1,c+2,x)M_2\biggr] + \frac{a(c-a)^2x^2}{(c+1)^2(c+2)}M(a+2,c+3,x)M(a+1,c+2,x)\\
    + \frac{(c-a)^2(c+a+1)x}{(c+1)(c+2)}M(a+2,c+3,x)M(a,c+1,x).
  \end{multline}
  However, on invoking the contiguous relation~\eqref{eq:8}, $h(x)$ can be
  simplified considerably to yield
  \begin{equation}
    \label{eq:21}
    \frac{(c+1)h(x)}{(c-a)^2x} = \frac{(a+1)(c+1)}{c+2}M(a,c+1,x)M(a+2,c+3,x)-aM(a+1,c+2,x)^2.
  \end{equation}
  Therefore, the condition $h(x) > 0$ is equivalent to (after introducing the notation $c' =
  c+1$)
  \begin{equation}
    \label{eq:29}
    \frac{c'}{c'+1}M(a,c',x)M(a+2,c'+2,x)-\frac{a}{a+1}M(a+1,c'+1,x)^2 > 0,
  \end{equation}
  or, in other words
  \begin{equation*}
    \frac{(a+1)}{(c'+1)}\frac{M(a+2,c'+2,x)}{M(a+1,c'+1,x)} > \frac{a}{c'}\frac{M(a+1,c'+1,x)}{M(a,c',x)}.
  \end{equation*}
  But this final inequality follows from Theorem~\ref{thm:monotone} by using $\mu=1$ in~(\ref{eq:13}).

  \emph{Upper-bound.} To prove the upper bound, first for brevity introduce
  the notation $b=c-a$, and $q=1-r$. The lower-bound in~(\ref{eq:xboundspos}) can be
  then rewritten as ($b-cq=cr-a > 0$)
  \begin{equation*}
    \frac{b-cq}{q(1-q)}\left(1+ \frac{q}{b}\right) < x,
  \end{equation*}
  which in turn can be rearranged to
  \begin{equation*}
    q^2(x-c/b) - q(x+c-1) + b < 0.
  \end{equation*}
  Note first that the equation $q^2(x-c/b)-q(x+c-1)+b=0$  has two distinct real roots since the discriminant (upon using $b=c-a$)
  \begin{equation*}
    D=(x+c-1)^2-4(c-a)(x-c/(c-a))=(1-x-c)^2+4ax+4c(1-x)=(1-x+c)^2+4ax>0.
  \end{equation*}
  We need to consider three cases:
  \begin{enumerate}[(1)]
  \item $x-c/b>0$ implies $q$ lies between the roots
    \begin{equation*}
      \frac{b(x+c-1)}{2(bx-c)}-\frac{b\sqrt{D}}{2(bx-c)} < q < \frac{b(x+c-1)}{2(bx-c)} + \frac{b\sqrt{D}}{2(bx-c)}.
    \end{equation*}
  \item $x-c/b<0$ implies $q$ is smaller than the smaller root or bigger than the bigger root, i.e.,
    \begin{equation*}
      q < \frac{b(x+c-1)}{2(bx-c)}+\frac{b\sqrt{D}}{2(bx-c)},\quad\text{or}\quad
      \frac{b(x+c-1)}{2(bx-c)}-\frac{b\sqrt{D}}{2(bx-c)}<q.
  \end{equation*}
  \item $x-c/b=0$ implies  $b/(x+c-1) < q$.
  \end{enumerate}
\noindent Since
\[
\lim\limits_{x\to{c/b}}\left(\frac{b(x+c-1)}{2(bx-c)}-\frac{b\sqrt{D}}{2(bx-c)}\right)=\frac{b}{x+c-1},
\]
in all three situations we have
\[
\frac{b(x+c-1)-b\sqrt{D}}{2(bx-c)} < q.
\]
Changing $a{\to}a+1$ and $c{\to}c+1$ here (recall that $b=c-a$) we
get
\begin{equation}
\label{eq:10}
0 < \frac{b(x+c)-b\sqrt{(x+c)^2-4(bx-c-1)}}{2(bx-c-1)} < 1-r_1,
\end{equation}
where as before $r_1=g(a+1,c+1,x)$. The positivity is clear on inspection for both $bx-c-1 > 0$ and $bx-c-1 < 0$.  Next, after suitably rewriting~(\ref{eq:2x}), we have
\[
x=\frac{b-cq}{(1-q)(1-r_1)}.
\]
Applying inequality~\eqref{eq:10} here, we obtain
\[
 x < \frac{2(b-cq)(bx-c-1)}{(1-q)b(x+c-\sqrt{(x+c)^2-4(bx-c-1)})}.
\]
Squaring and simplifying we get the inequality
\[
(bx-c-1)(x^2q(1-q)b(c-b)-xb(b-cq)(c-b)-(c+1)(b-cq)^2)<0
\]
for $bx-c-1>0$ and the inequality
\[
(bx-c-1)(x^2q(1-q)b(c-b)-xb(b-cq)(c-b)-(c+1)(b-cq)^2)>0
\]
for $bx-c-1<0$.  Hence both situations reduce to the single inequality
\[
x^2q(1-q)b(c-b)-xb(b-cq)(c-b)-(c+1)(b-cq)^2<0,
\]
which on plugging $q=1-r$ and $b=c-a$ becomes
\[
x^2r(1-r)a(c-a)-xa(c-a)(cr-a)-(c+1)(cr-a)^2 < 0.
\]
Since the coefficient at $x^2$ is clearly positive $x$ must lie
between the roots, in particular it should be smaller than the
bigger root, which is the upper bound in~\eqref{eq:xboundspos}.

\emph{The rightmost bound.} Verifying the rightmost inequality is
an exercise in high-school algebra.
\end{proof}

\begin{theorem}[Negative $\kappa$]\label{thm.negkappa}
  \label{thm:negkappa}
  Let $\kappa(r)$ be the solution to~(\ref{eq:61}), $c > a > 0$, and $r$ lie in $(0, a/c)$. Then, we have the following bounds:
  \begin{equation}
    \label{eq:xboundsnegative}
    \frac{rc-a}{r(1-r)}\left(1+\frac{1-r}{c-a}\right) < \frac{rc-a}{2r(1-r)}\left(1+\sqrt{1+\frac{4(c+1)r(1-r)}{a(c-a)}}\right) < \kappa(r) < \frac{rc-a}{r(1-r)}\left(1+\frac{r}{a}\right)
  \end{equation}
\end{theorem}
\begin{proof}
  \emph{Upper bound:} To simplify notation we write as before $x=\kappa(r)$. 
Recall that $r_1=g(a+1,c+1;x)$; according to (\ref{eq:rr1}), the
value $r_1>r$; so in view of $cr-a<0$ and~\eqref{eq:2x}, we have
the inequality
\begin{equation*}
  x<\frac{cr-a}{r(1-r)},
\end{equation*}
which is equivalent to (noting that $x<0$)
\begin{equation*}
  r^2+\left(\frac{c}{x}-1\right)r-\frac{a}{x}<0.
\end{equation*}
Thus, $r$ must lie between the roots of the quadratic
\begin{equation*}
  r^2+\left(\frac{c}{x}-1\right)r-\frac{a}{x}=0.
\end{equation*}
Straightforward analysis shows that the discriminant is positive for all real $x$ if $c>a>0$, and we have two distinct real roots so that
\begin{equation}\label{eq:r-quadratic}
\frac{1}{2}-\frac{c}{2x}-\frac{1}{2}\sqrt{1+\frac{2(2a-c)}{x}+\frac{c^2}{x^2}}<r<\frac{1}{2}-\frac{c}{2x}+\frac{1}{2}\sqrt{1+\frac{2(2a-c)}{x}+\frac{c^2}{x^2}}.
\end{equation}
We will use the lower bound above written here for $r_1=g(a+1,c+1;x)$
\begin{equation*}
0<\frac{1}{2}-\frac{c+1}{2x}-\frac{1}{2}\sqrt{1+\frac{2(2a-c+1)}{x}+\frac{(c+1)^2}{x^2}}<r_1.
\end{equation*}
Applying this lower bound for $r_1$ to~\eqref{eq:2x} and dividing by $2x<0$, we obtain the bound
\[
x<\frac{2(cr-a)}{r(1+(c+1)/x+\sqrt{1+2(2a-c+1)/x+(c+1)^2/x^2})}.
\]
By high school algebra we immediately conclude that the denominator is positive, so that
\[
(-x)\sqrt{1+2(2a-c+1)/x+(c+1)^2/x^2}<(x+1-c)+2a/r
\]
Squaring and rearranging we obtain the desired upper bound in (\ref{eq:xboundsnegative}).

\emph{Lower bounds:} First we prove the leftmost bound in (\ref{eq:xboundsnegative}) (i.e., we show that it less than
$x$). Since $r \in (0,a/c)$ we have $cr-a<0$; so, following the line of proof of Theorem~\ref{thm:poskappa} we get
\[
\frac{(c-a-1)r+1}{c-a+1-r}>r_1
\]
which leads to $h(x)<0$ where $h$ is as defined in the course of proof of Theorem~\ref{thm:poskappa}.  However, since $x<0$ for $r\in(0,a/c)$ we must again show that (where $c' = c + 1$)
\begin{equation*}
\frac{(a+1)M(a+2;c'+2;x)}{(c'+1)M(a+1;c'+1;x)}>\frac{aM(a+1;c'+1;x)}{c'M(a;c';x)},
\end{equation*}
but this time for $x<0$. Applying the Kummer transformation to the inequality above leads to
\begin{equation*}
\frac{[\Gamma(a+\mu+1)/\Gamma(c+\mu+1)]M(c-a;c+\mu+1;y)}{[\Gamma(a+\mu)/\Gamma(c+\mu)]M(c-a;c+\mu;y)}>
\frac{[\Gamma(a+1)/\Gamma(c'+1)]M(c'-a;c+1;y)}{[\Gamma(a)/\Gamma(c')]M(c'-a;c';y)},
\end{equation*}
with $\mu=1$ and $y=-x>0$. This inequality follows immediately from Theorem~\ref{lemm:lcvx2}  as we have demonstrated in the proof of the previous theorem.

Proving the second (and tighter) lower bound in (\ref{eq:xboundsnegative}) requires more work. Introduce thus the notation $b=c-a$ and $q=1-r$. The leftmost  bound in (\ref{eq:xboundsnegative}) can be rewritten as
\begin{equation*}
  \frac{b-cq}{q(1-q)}\left(1+\frac{q}{b}\right)<x,
\end{equation*}
which can be rearranged to
\[
q^2(x-c/b)-q(x+c-1)+b<0.
\]
The discriminant of the corresponding quadratic equation equals
\begin{equation*}
  D=(x+c-1)^2-4(c-a)(x-c/(c-a))=(x+c-1)^2-4(c-a)x+4c>0,
\end{equation*}
since $c-a>0$ and $x<0$. For $x-c/b<0$ this implies that $q$ must be
smaller than the smaller root or bigger than the bigger root of
$q^2(x-c/b)-q(x+c-1)+b=0$. That is,
\begin{equation*}
q<\frac{b(x+c-1)}{2(bx-c)}+\frac{b\sqrt{(x+c-1)^2-4(bx-c)}}{2(bx-c)},
\end{equation*}
or
\begin{equation*}
  \frac{b(x+c-1)}{2(bx-c)}-\frac{b\sqrt{(x+c-1)^2-4(bx-c)}}{2(bx-c)}<q.
\end{equation*}
Changing $a{\to}a+1$ and $c{\to}c+1$ here (recall that $b=c-a$), from the second inequality we obtain
\[
0<\frac{b[(x+c)-\sqrt{(x+c)^2-4(bx-c-1)}]}{2(bx-c-1)}<1-r_1,
\]
where as before $r_1=g(a+1,c+1,x)$. Positivity follows by separately considering $x+c\geq{0}$ and $x+c<0$.  Next, a simple manipulation of~\eqref{eq:2x} shows that
\begin{equation*}
  x=\frac{b-cq}{(1-q)(1-r_1)}.
\end{equation*}
Applying the above inequality concerning $r_1$ here, we obtain the bound (since $b-cq<0$)
\begin{equation*}
  x>\frac{2(b-cq)(bx-c-1)}{b(1-q)[(x+c)-\sqrt{(x+c)^2-4(bx-c-1)}]}.
\end{equation*}
Here $b-cq<0$ (since $r<a/c$), $bx-c-1<0$ (since $x<0$, $b>0$) and $(x+c)-\sqrt{(x+c)^2-4(bx-c-1)}<0$ (since the second
term is bigger than the first). So the above bound on $x$ is negative as expected.  By simple algebra and in view of the signs we have
\[
x+c-\frac{2(b-cq)(bx-c-1)}{xb(1-q)}>\sqrt{(x+c)^2-4(bx-c-1)}>0.
\]
Squaring yields
\begin{equation*}
  \begin{split}
    &x^2b^2(1-q)^2(x+c)^2 - 4(b-cq)(bx-c-1)xb(1-q)(x+c)+4(b-cq)^2(bx-c-1)^2\\
    &>((x+c)^2-4(bx-c-1))x^2b^2(1-q)^2,
  \end{split}
\end{equation*}
whereby reducing similar terms and on dividing by $-4(bx-c-1)>0$ we obtain
\begin{equation*}
  (b-cq)xb(1-q)(x+c)-(b-cq)^2(bx-c-1)-x^2b^2(1-q)^2>0.
\end{equation*}
Simplifying we get
\[
x^2q(1-q)b(c-b)-xb(b-cq)(c-b)-(c+1)(b-cq)^2<0,
\]
so that after plugging in $q=1-r$ and $b=c-a$ we have the inequality
\[
x^2r(1-r)a(c-a)-xa(c-a)(cr-a)-(c+1)(cr-a)^2<0.
\]
Thus, $x$ must be between the roots of this quadratic. The fact that it is bigger than the smallest root is precisely the inequality between $x$ and the middle term in (\ref{eq:xboundsnegative}).

Finally, comparing the two lower bounds in (\ref{eq:xboundsnegative}) leads to the inequality
\[
\frac{1}{2}+\frac{1-r}{c-a}>\frac{1}{2}\sqrt{1+\frac{4(1-r)r(c+1)}{(c-a)a}},
\]
which upon squaring and simplifying reduces to $r<a/c$, the hypothesis of the theorem.~~$\square$
\end{proof}

In the next theorem  we find the asymptotic expansions of the
solution to $g(a,c;x)=r$ around $r=0$, $r=a/c$ and $r=1$.
\begin{theorem}\label{th:xasymp}
  Let $c > a > 0$, $r \in (0,1)$; let $x(r)$ be the solution to $g(a, c; x)=r$. Then,
  \begin{align}
    \label{eq:xasymp.0}
    x(r) &= -\frac{a}{r}+(c-a-1)+\frac{(c-a-1)(1+a)}{a}r+O(r^2),\quad r\to{0},\\
    \label{eq:xasymp.ac}
    x(r) &= \Bigl(r-\frac{a}{c}\Bigr) \left  \{ \frac{c^2(1+c)}{a(c-a)} + \frac{c^3(1+c)^2(2a-c)}{a^2(c-a)^2(c+2)} \Bigl(r-\frac{a}{c}\Bigr) + O\Bigl(\Bigl(r-\frac{a}{c}\Bigr)^2\Bigr) \right\},\quad r\to{\frac{a}{c}}\\
\label{eq:xasymp.1}
    x(r) &= \frac{c-a}{1-r} + 1-a + \frac{(a-1)(a-c-1)}{c-a}(1-r) + O((1-r)^2),\quad r\to{1}.
  \end{align}
\end{theorem}
\begin{proof}  The first step
is the standard division of power series (see, for
instance, \cite[Theorem 8.8]{Gonzalez} or
\cite[Chapter~2.3]{HurwKur}) which yields in a neighbourhood of
$x=0$:
\begin{equation}\label{eq:f-0}
g(x)=\frac{aM(a+1,c+1;x)}{cM(a,c;x)}=\frac{a}{c}+\frac{a(c-a)}{c^2(1+c)}x+\frac{a(a-c)(2a-c)}{c^3(1+c)(2+c)}x^2+O(x^3).
\end{equation}
Next, the division of the asymptotic formulas
(\ref{eq:kummer.asym}) used with appropriate parameters gives in
the neighbourhood of $x=\infty$:
\begin{equation}\label{eq:f-infty}
g(x)=\frac{aM(a+1,c+1;x)}{cM(a,c;x)}=1+(a-c)\frac{1}{x}+(a-c)(1-a)\frac{1}{x^2}+(1-a)(c-a)(2a-c-2)\frac{1}{x^3}+O(x^{-4}).
\end{equation}
The inversion of a power series in a neighbourhood of a finite
point $x_0$ is achieved via the following formula of Lagrange
\cite[Chapter~7, Theorem~1*]{HurwKur}:
\[
x(r)=x_0+\sum\limits_{n=1}^{\infty}\lim\limits_{x\to{x_0}}\left[\frac{d^{n-1}}{dx^{n-1}}\left(\frac{x-x_0}{g(x)-r_0}\right)^n\right]
\frac{(r-r_0)^n}{n!},
\]
where $r_0=g(x_0)$ and $g'(x_0)\ne{0}$.  Introducing the notation
\[
g(x)-r_0=g_1(x-x_0)+g_2(x-x_0)^2+g_3(x-x_0)^3+\cdots
\]
for the Taylor coefficients of $g$, we obtain (see
\cite[(8.14.11)]{Gonzalez}):
\begin{equation}\label{eq:inversion}
x(r)-x_0=\frac{1}{g_1}(r-r_0)-\frac{g_2}{g_1^3}(r-r_0)^2+\frac{2g_2^2-g_1g_3}{g_1^5}(r-r_0)^3+O((r-r_0)^4).
\end{equation}
In our case we have expansions in the neighbourhoods of $x_0=0$
and $x_0=\infty$.  For the case $x_0=0$ formula
(\ref{eq:inversion}) immediately reveals:
\begin{equation}\label{eq:xrac}
x(r)=(r-a/c)\left\{\frac{c^2(1+c)}{a(c-a)}+\frac{c^3(1+c)^2(2a-c)}{a^2(c-a)^2(c+2)}(r-a/c)+O((r-a/c)^2)\right\}
\end{equation}
which is precisely formula (\ref{eq:xasymp.ac}). For the point at
infinity the Lagrange formula does not have the form given above.
To compute the correct expression introduce the new variable
$y=1/x$ and rewrite the expansion (\ref{eq:f-infty}) in the form:
\[
r=\tilde{f}(y)=f(1/y)=q_0+q_1y+q_2y^2+\cdots
\]
Then according  to (\ref{eq:inversion}):
\[
y=\frac{1}{q_1}(r-q_0)-\frac{q_2}{q_1^3}(r-q_0)^2+\frac{2q_2^2-q_1q_3}{q_1^5}(r-q_0)^3+O((r-q_0)^4).
\]
Again applying standard division of power series (see, for
instance \cite[Theorem 8.8]{Gonzalez} or
\cite[Chapter~2.3]{HurwKur}) we get:
\[
x(r)=\frac{1}{y(r)}=\frac{q_1}{r-q_0}+\frac{q_2}{q_1}+\frac{1}{q_1}\left(\frac{q_3}{q_1}-\frac{q_2^2}{q_1^2}\right)(r-q_0)
+O((r-q_0)^2)
\]
Substituting here the values of $q_i$ from (\ref{eq:f-infty}), we
obtain the expansion:
\begin{equation}\label{eq:xr1}
x(r)=\frac{1}{1-r}\left\{(c-a)+(1-a)(1-r)+\frac{(a-1)(a-c-1)}{c-a}(1-r)^2+O((1-r)^3)\right\}.
\end{equation}
This proves formula (\ref{eq:xasymp.1}).

From formula (\ref{eq:fmuneg.asymp}) we immediately derive
$$
g(a;c;x)=\frac{a}{-x}+\frac{a(1+a-c)}{(-x)^2}+\frac{a(1+a-c)(2+2a-c)}{(-x)^3}+O(1/(-x)^4)
$$
Introducing $y=-1/x$ we rewrite this as
$$
r=ay+a(1+a-c)y^2+a(1+a-c)(2+2a-c)y^3+O(y^4)=q_0+q_1y+q_2y^2+q_3y^3+O(y^4)
$$
Then according  to (\ref{eq:inversion}):
\[
y=\frac{1}{q_1}(r-q_0)-\frac{q_2}{q_1^3}(r-q_0)^2+\frac{2q_2^2-q_1q_3}{q_1^5}(r-q_0)^3+O((r-q_0)^4).
\]
Again applying standard division of power series (see, for
instance \cite[Theorem 8.8]{Gonzalez} or
\cite[Chapter~2.3]{HurwKur}) we get:
\[
x(r)=\frac{-1}{y(r)}=-\frac{q_1}{r-q_0}-\frac{q_2}{q_1}-\frac{1}{q_1}\left(\frac{q_3}{q_1}-\frac{q_2^2}{q_1^2}\right)(r-q_0)
+O((r-q_0)^2)
\]
or
$$
x(r)=\frac{-1}{y(r)}=-\frac{a}{r}-\frac{a(1+a-c)}{a}-\frac{1}{a}\left(\frac{a(1+a-c)(2+2a-c)}{a}-\frac{a^2(1+a-c)^2}{a^2}\right)r+O(r^2)
$$
$$
=-\frac{a}{r}+(c-a-1)+\frac{(c-a-1)(1+a)}{a}r+O(r^2),
$$
which completes the demonstration of (\ref{eq:xasymp.0}).
\end{proof}
\end{document}